\documentclass{svjour3}                     

\smartqed  
\usepackage{enumerate}

\usepackage{amsmath}
\usepackage{amsfonts}
\usepackage{amssymb}
\usepackage{amsopn}
\usepackage{latexsym}
\usepackage{units}

\usepackage{mathrsfs}

\usepackage{xy}

\providecommand{\dd}{\mathrm{d}} \providecommand{\tcz}{\diamond_0}
\providecommand{\tcB}{\diamond_{\hbar}^B}
\providecommand{\order}{\mathcal{O}}

\def\A{{\mathcal A}}
\def\k{\kappa}
\def\o{\omega}
\def\O{\Omega}
\def\th{\theta}

\def\Or{\mathcal O}
\def\N{\mathbb{N}}
\def\R{\mathbb{R}}
\def\T{\mathbb{T}}

\def\X{\mathscr X}

\def\QQ{\mathbf Q}

\def\C{\mathbb{C}}

\def\H{\mathcal H}

\def\Q{\mathcal Q}

\def\S{\mathcal S}

\def\U{\mathcal U}
\def\CCC{{\mathfrak C}}

\def\Rep{\mathfrak{Rep}}

\def\de{\mathrm{d}}

\def\I{{\rm 1\kern-.26em I}}

\def\BB{\mathfrak B}

\def\1{\mathfrak{1}}
\def\0{\mathfrak{0}}

 \def\hb{\hbar}

\def\Op{\mathfrak{Op}}

\providecommand{\CCC}{\mathfrak{C}}

 \providecommand{\abs}[1]{\left
| #1 \right |} \providecommand{\sabs}[1]{| #1 \vert}
\providecommand{\babs}[1]{\bigl | #1 \bigr |}
\providecommand{\Babs}[1]{\Bigl | #1 \Bigr |}
\providecommand{\norm}[1]{\left \Vert #1 \right \Vert}
\providecommand{\snorm}[1]{\Vert #1 \Vert}
\providecommand{\bnorm}[1]{\bigl \Vert #1 \bigr \Vert}
\providecommand{\Bnorm}[1]{\Bigl \lVert #1 \Bigr \rVert}

\providecommand{\expval}[1]{\left \langle #1 \right \rangle}

\providecommand{\ie}{i.~e.}

\begin{document}

\title{Magnetic twisted actions on general abelian $C^*$-algebras}

\author{Fabian Belmonte         \and
         Max Lein             \and
     Marius M\u antoiu
}


\institute{Fabian Belmonte \at
              Departamento de Matem\'aticas, Universidad de Chile, Las Palmeras 3425, Casilla 653,
Santiago, Chile \\
              \email{fabianbelmonte@gmail.com}           
           \and
           Max Lein \at
             Zentrum Mathematik, Technische Universität München, Boltzmannstrasse 3, D-85748 Garching,
             Germany\\
             \email{lein@ma.tum.de}           
           \and
            Marius M\u antoiu \at
              Departamento de Matem\'aticas, Universidad de Chile, Las Palmeras 3425, Casilla 653,
Santiago, Chile \\
              Tel.: +56-02-9787300\\
              \email{mantoiu@uchile.cl}           
}

\date{Received: date / Accepted: date}

\maketitle

\begin{abstract}
We introduce magnetic twisted actions of $\X=\R^n$ on general
abelian $C^*$-algebras and study the associated twisted crossed
product and pseudodifferential algebras in the framework of strict
deformation quantization.\\


\end{abstract}

\section{Introduction}
\label{duci} The usual pseudodifferential calculus in phase space
$\Xi:=T^*\R^n$ is connected to crossed product $C^*$-algebras $\A
\rtimes_\th \X$ associated to the action by translations $\th$ of
the group $\X:=\R^n$ on an abelian $C^*$-algebra $\A$ composed of
functions defined on $\X$. Such a formalism has been used in the
quantization of a physical system composed of a spin-less particle
moving in $\X$, where the operators acting on $L^2(\X)$ can be
decomposed into the building block observables position and
momentum which are associated to $\X$ and its dual $\X^*$. When
dealing with Hamiltonian operators, the algebra $\A$ encapsulates
properties of electric potentials, for instance.

During the last decade, it was shown how to incorporate correctly
a variable magnetic field in the picture, cf.
\cite{Mu,KO1,KO2,MP1,MP3,MPR1,MPR2,IMP1,IMP2} (see also
\cite{BB1,BB2,BB3} for extensions involving nilpotent groups).
This relies on twisting both the pseudodifferential calculus and
the crossed product algebras by a $2$-cocycle defined on the group
$\X$ and taking values in the (Polish, non-locally compact group)
$\mathcal U(\A)$ of unitary elements of the algebra $\A$. This
$2$-cocycle is given by imaginary exponentials of the magnetic
flux through triangles. The resulting gauge-covariant formalism
has position and \emph{kinetic} momentum as its basic observables.
The latter no longer commute amongst each other due to the
presence of the magnetic field. It was shown in \cite{MP2} that
the family of twisted crossed products indexed by $\hbar \in
(0,1]$ can be understood as a strict deformation quantization (in
the sense of Marc Rieffel) of a natural Poisson algebra defined by
a symplectic form which is the sum of the canonical symplectic
form in $\Xi$ and a magnetic contribution.

A natural question is what happens when the algebra $\A$ (composed
of functions defined on $\X$) is replaced by a general abelian
$C^*$-algebra. By Gelfand theory this one is isomorphic to
$C_0(\O)$, the $C^*$-algebra of all the complex continuous
functions vanishing at infinity defined on the locally compact
space $\O$. To define crossed products and pseudodifferential
operators we also need a continuous action $\th$ of $\X$ on $\O$
by homeomorphisms. $C_0(\O)$ can be seen as a $C^*$-algebra of
functions on $\X$ exactly when $\O$ happens to have a
distinguished dense orbit. In the general case, the twisting
ingredient will be ``a general magnetic field'', i.e. a continuous
family $B$ of magnetic fields indexed by the points of $\O$ and
satisfying an equivariance condition with respect to the action
$\th$.

The purpose of this article is to investigate the emerging
formalism, both classical and quantal.

To the quadruplet $(\O,\th,B,\X)$ described above we first assign
in Section \ref{secintro} a Poisson algebra that is the setting
for classical mechanics. The Poisson bracket is written with
derivatives defined by the abstract action $\th$ and it also
contains the magnetic field $B$. Since $\O$ does not have the
structure of a manifold, this Poisson algebra does not live on a
Poisson manifold, let alone a symplectic manifold (as it is the
case when a dense orbit exists). But it admits symplectic
representations and, at least in the free action case,
$\O\times\X^*$ is a Poisson space \cite{La} in which symplectic
manifolds (the orbits of the action raised to the phase-space
$\Xi$) are only glued together continuously.

Twisted crossed product $C^*$-algebras are available in a great
generality \cite{PR1,PR2}. We use them in Section \ref{clll} to
define algebras of quantum observables with magnetic fields. By a
partial Fourier transformation they can be rewritten as algebras
of generalized magnetic pseudodifferential symbols. The outcome
has some common points with Rieffel's pseudodifferential calculus
\cite{Rie1}, which starts from an action of $\R^N$ on a
$C^*$-algebra. In our case this algebra is abelian and the action
has a somehow restricted form; on the other hand the magnetic
twisting cannot be covered by Rieffel's formalism. We also study
Hilbert-space representations of the algebras of symbols. Their
interpretation as equivariant families of usual magnetic
pseudodifferential operators with anisotropic coefficients
\cite{LMR} is available. This will be developed in a forthcoming
article and applied to spectral analysis of deterministic and
random magnetic quantum Hamiltonians.

Section 4 is dedicated to a development of the magnetic
composition law involving Planck's constant. The first and second
terms are written using the classical Poisson algebra conterpart.
We insist on reminder estimates valid in the relevant $C^*$-norms.

All these are used in Section 5 to show that the quantum formalism
converges to the classical one when Planck's constant $\hb$
converges to zero, in the sense of strict deformation quantization
\cite{Rie1,Rie2,La,La1}. The semiclassical limit of dynamics
\cite{La,Rie3} generated by generalized magnetic Hamiltonians
will be studied elsewhere.

An appendix is devoted to some technical results about the
behavior of the magnetic flux through triangles. These results are
used in the main body of the text.

\section{Classical}\label{secintro}

\subsection{Actions}\label{dinas}

Let $\A$ denote an abelian $C^*$-algebra. By Gelfand theory, this algebra is isomorphic to the algebra
$C_0(\O)$ of continuous functions vanishing at infinity on some locally compact (Hausdorff) topological space $\O$,
and we shall treat this isomorphism as an identification. Furthermore, we shall always assume that $\A$ is endowed with a
continuous action $\th$ of the group $\X:=\R^n$ by automorphisms: For any $x,y \in \X$ and $\varphi\in \A$,
\begin{align*}
\theta_0 [\varphi] = \varphi ,
\qquad
\theta_x \bigl [ \theta_y [\varphi] \bigr ] = \theta_{x+y}[\varphi]
\end{align*}
and the map $\X \ni x \mapsto \theta_x[\varphi] \in \A$ is continuous for any $\varphi\in\A$. The triple $(\A,\th,\X)$
is usually called an (abelian) $\X$-{\it algebra}.

Equivalently, we can assume that the spectrum $\O$ of $\A$ is endowed with a continuous action of $\X$ by
homeomorphisms, which with abuse of notation will also be denoted by $\th$. In other words, $(\O,\th,\X)$ is a locally compact
dynamical system. We shall use all of the notations $\th(\o,x)=\th_x[\o]=\th_{\o}(x)$
for $(\o,x)\in\O\times\X$ and choose the convention
$\bigl ( \theta_x[\varphi] \bigr )(\omega) = \varphi \bigl ( \theta_x[\omega] \bigr )$ to connect the two actions.

An important, but very particular family of examples of $\X$-algebras is constructed using functions on $\X$. We denote by
$BC(\X)$ the $C^*$-algebra of all bounded, continuous functions $\phi:\X \longrightarrow \C$. Let $\tau$ denote the action
of the locally compact group $\X = \R^n$ on itself, {\it i.e.}~for any $x,y \in \X$ we set $\tau(x,y)=\tau_x[y]:=y+x$.
This notation is also used for the action of $\X$ on $BC(\X)$ given by $\tau_x[\varphi](y):= \varphi(y + x)$. The action is
continuous only on $BC_{{\rm u}}(\X)$, the $C^*$-subalgebra composed of bounded and uniformly continuous functions.
Any $C^*$-subalgebra of $BC_{{\rm u}}(\X)$ which is invariant under translations is an $\X$-algebra.
Motivated by the above examples, we define
$
BC(\O) :=\{\varphi:\O\rightarrow\mathbb C\mid f\ {\rm is\ bounded\ and\ continuous}\}
$
and
\begin{align*}
    \mathcal{B} \equiv BC_{{\rm u}}(\O) := \bigl \{ \varphi\in BC(\O) \mid \X \ni x \mapsto \theta_x[\varphi] \in BC(\O)
    \mbox{ is continuous} \bigr \} .
\end{align*}
%

\medskip
By a \emph{$\X$-morphism} we denote either a continuous map between  the underlying spaces of two dynamical systems
which intertwines the respective actions, or a morphism between two $\X$-algebras which also intertwines their
respective actions.

Let us recall some definitions related to the dynamical system $(\O,\th,\X)$. For any $\o\in \O$ we set
$\Or_\o := \bigl \{ \theta_x[\o] \mid x\in\X \bigr \}$ for \emph{the orbit of} $\o$ and $\Q_\o:=\overline{\Or_\o}$ for
\emph{the
quasi-orbit of} $\o$, which is the closure of $\Or_\o$ in $\O$. We shall denote by $\mathbf O(\O)\equiv \mathbf
O(\O,\th,\X)$ the set of orbits of $(\O,\th,\X)$ and by $\QQ(\O)\equiv \QQ(\O,\th,\X)$ the set of quasi-orbits of
$(\O,\th,\X)$. For fixed $\o \in \O$, $\varphi \in C_0(\O)$ and $x \in \X$, we set
$\varphi_\o(x):= \varphi(\theta_x[\o])\equiv \varphi\big(\th_\o(x)\big)$.
It is easily seen that $\varphi_\o: \X\to \C$ belongs to $BC_u(\X)$. Furthermore, the $C^*$-algebra
$$
\A_{\o} := \bigl \{\varphi_\o\mid \varphi \in C_0(\O) \bigr \} = \th_{\o}[C_0(\O)]
$$
is isomorphic to the $C^*$-algebra $C_0(\Q_\o)$ obtained by restricting the elements of $C_0(\O)$ to the closed invariant subset
$\Q_\o$. Then, one clearly obtains that
\begin{equation}\label{titi}
\th_\o: C_0(\O)\ni \varphi \mapsto \varphi_\o=\varphi\circ\th_\o \in BC_u(\X)
\end{equation}
is a $\X$-morphism between $\big(C_0(\O),\th,\X\big)$ and $\big(BC_u(\X),\tau,\X\big)$ which induces a
$\X$-isomorphism between $\big(C_0(\Q_\o),\th,\X\big)$ and $\big(\A_\o,\tau,\X\big)$.

We recall that the dynamical system is \emph{topologically transitive} if an orbit is dense, or equivalently if
$\O\in\QQ(\O)$. This happens exactly when the morphism (\ref{titi}) is injective for some $\o$. The dynamical system
$(\O,\th,\X)$ is \emph{minimal} if all the orbits are dense, {\it i.e.}~$\QQ(\O)=\{\O\}$. This property is also
equivalent to the fact that the only closed invariant subsets are $\emptyset$ and $\O$.

\begin{definition}\label{smooth}
Let ($\A,\th,\X)$ be an $\X$-algebra. We define the spaces of smooth vectors
\begin{align*}
\A^\infty:=\{\varphi\in\A\mid \X\ni x\mapsto\th_x(\varphi)\in\A\ {\rm is}\ C^\infty\}.
\end{align*}
\end{definition}

For the $\X$-algebras $C_0(\O)$ and $BC_{{\rm u}}(\O)$ we will often use the notations $ C_0^\infty(\O)$, respectively.
Despite these notations, we stress that in general $\Omega$ is not a manifold; the notion of differentiability is defined only
along orbits. By setting for any $\alpha\in\N^n$
\begin{equation*}
\delta^\alpha : C_0^\infty(\O)\rightarrow C_0^\infty(\O),\ \ \ \ \
\delta^\alpha\varphi:=\partial^\alpha_x\left(\varphi\circ\th_x\right)|_{x=0},
\end{equation*}
one defines a Fr\'echet structure on $C_0^\infty(\O)$ by the semi-norms
\begin{equation*}
s^\alpha(\varphi):= \bnorm{\delta^\alpha\varphi}_{C_0(\O)} = \sup_{\o \in \O} \babs{(\delta^\alpha \varphi)(\o)}.
\end{equation*}
Each of the two spaces, $C_0^\infty(\O)$ and $\A_\o^\infty$, is a dense Fr\'echet $^*$-subalgebra of the corresponding
$C^*$-algebra.

\begin{lemma}\label{rifi}
\begin{enumerate}[(i)]
\item For each $\o\in\O$ one has
$$
\A_\o^\infty= \bigl \{\phi\in C^\infty(\X)\mid\partial^\beta\phi\in\A_{\o},\ \forall \beta\in\mathbb N^n \bigr \}.
$$
In particular $\,\A_\o^\infty\subset BC^\infty(\X):= \bigl \{ \phi\in C^\infty(\X)\mid \partial^\beta\phi\ {\rm is\ bounded}\,\
\forall \beta\in \N^n \bigr \}$.
\item Let $\varphi\in C_0(\O)$. Then
$$
\varphi\in C_0^\infty(\O)\ \Longleftrightarrow\ \varphi\circ\th_\o\in\A^\infty_{\o},\ \forall\o\in\O.
$$
\end{enumerate}
\end{lemma}

\begin{proof}
The proof consists in some routine manipulations of the definitions. The only slightly non-trivial fact is to show that
point-wise derivations are equivalent to the uniform ones, required by the uniform norms. This follows from the
Fundamental Theorem of Calculus, using the higher order derivatives, which are assumed to be bounded.
A model for such a standard argument is the proof of Lemma 2.7 in \cite{LMR}.
\end{proof}
\begin{remark}
In the following, we will use repeatedly and without further comment the identification of point-wise and uniform derivatives
under the assumption that higher-order point-wise derivatives exist and are bounded.
\end{remark}
Although in our setting the classical observables are functions defined on $\O\times\X^*$, we are going to relate them to
functions on \emph{phase space} $\Xi:=\X\times\X^*$ whose points are denoted by capital letters $X=(x,\xi),\,Y=(y,\eta),
\,Z=(z,\zeta)$. The dual space $\X^*$ also acts on itself by translations: $\tau^*_\eta(\xi):=\xi+\eta$, and this action
is raised to various function spaces as above. Similarly, phase space $\Xi$ can also be regarded as a group acting on
itself by translations,  $(\tau\otimes\tau^*)_{(y,\eta)}(x,\xi):=(x+y,\xi+\eta)$. Phase space $\Xi$ acts on
$\Omega \times \X^*$ as well, via the action $\theta \otimes \tau^*$, and this defines naturally function spaces on
$\Omega\times\X^*$ as above; they will be used without further comment.

\subsection{Cocycles and magnetic fields}\label{coci}

We first recall the definition of a $2$-cocycle $\k$ on the abelian algebra $\A=C_0(\O)$ endowed with an action $\theta$
of $\X$. We mention that the group $\U \bigl ( \A \bigr )$ of unitary elements of the unital $C^*$-algebra $BC(\O)$
coincides with $C(\O;\T):= \bigl \{\varphi\in C(\O)\mid |\varphi(\o)|=1,\ \forall \o\in\O \bigr \}$,
on which we consider the topology of uniform convergence on compact sets.

\begin{definition}\label{canon}
{\rm A normalized $2$-cocycle
on} $\A$ is a continuous map $\k:\X \times \X \to \U(\A)$ satisfying for all $x,y,z \in \X$:
\begin{equation}\label{2coc}
\kappa(x+y,z)\;\!\kappa(x,y) = \th_x[\kappa(y,z)]\;\!\kappa(x,y+z)
\end{equation}
and $\kappa(x,0) = \kappa(0,x)=1$.
\end{definition}

\begin{proposition}\label{zisan}
If $\k:\X\times\X\rightarrow C(\O;\T)$ is a $2$-cocycle of $C_0(\O)$ then for any $\o\in\O$,
$\k_\o(\cdot,\cdot):=\k(\cdot,\cdot)\circ\th_\o$ is a $2$-cocycle of $\A_\o$ with respect to the action $\tau$.
\end{proposition}
\begin{proof}
Everything is straightforward. To check the $2$-cocycle property,
one uses
$$
\theta_x\circ\theta_\o=\th_\o\circ\tau_x,\ \ \ \ \ x\in\X,\ \o\in\O.
$$
It is easy to show that $\k:\X\times\X\rightarrow C(\O,\mathbb T)$ is continuous iff the function
$$
\O\times\X\times\X\ni(\o,x,y)\mapsto\k(\o;x,y):= \bigl (\k(x,y) \bigr )(\o)\in\T
$$
is continuous. Recalling the isomorphism $\A_\o\cong C(\mathcal Q_\o)$ one easily finishes the proof.
\end{proof}

We shall be interested in {\it magnetic $2$-cocycles}.
\begin{definition}\label{magneti}
We call {\rm magnetic field on} $\O$ a continuous function
$B:\O\rightarrow \bigwedge^2\X$ such that $B_\o:=B\circ\th_\o$ is a magnetic field (continuous closed $2$-form on $\X$)
for any $\o$.
\end{definition}

Using coordinates, $B$ can be seen as an anti-symmetric matrix
$\left(B^{jk}\right)_{j,k}$ where the entries are continuous
functions $B^{jk}:\Omega\rightarrow\R$ satisfying (in the
distributional sense)
$$
\partial_j B_\o^{kl}+\partial_k B_\o^{lj}+\partial_l B_\o^{jk}=0,\ \ \ \ \ \forall\,\o\in\O,\ \,\forall\,j,k,l=1,\dots,n.
$$
\begin{proposition}\label{si}
Let $B$ a magnetic field on $\O$. Set
$$
\bigl (\k^B(x,y) \bigr )(\o)\equiv\k^B(\o;x,y):=\exp \bigl (-i\Gamma^{B_\o}\langle 0,x,x+y\rangle \bigr ) ,
$$
where $\Gamma^{B_\o}\langle a,b,c\rangle:=\int_{\langle a,b,c\rangle}B_\o$ is the integral (flux) of the $2$-form
$B_\o$ through the triangle $\langle a,b,c\rangle$ with corners $a,b,c\in\X$.
Then $\k^B$ is a $2$-cocycle on  $C_0(\O)$.
\end{proposition}

\begin{proof}
The algebraic properties follow from the properties of the integration of $2$-forms. For example, (\ref{2coc}) is a
consequence of the identity
$$
\Gamma^{B_\o}\langle 0,x,x+y\rangle+\Gamma^{B_\o}\langle 0,x+y,x+y+z\rangle=\Gamma^{B_{\th_x[\o]}}\langle 0,y,y+z\rangle+
\Gamma^{B_\o}\langle 0,x,x+y+z\rangle.
$$
This one follows from Stokes' Theorem, after noticing that
\begin{equation}\label{chec}
\Gamma^{B_{\th_x[\o]}}\langle 0,y,y+z\rangle=\Gamma^{B_{\o}}\langle x,x+y,x+y+z\rangle.
\end{equation}
One still has to check that $\k^B\in C(\O\times\X\times\X)$. This reduces to the obvious continuity of
$$
(\o,x,y)\mapsto\Gamma^{B_\o}\langle 0,x,x+y\rangle = \sum_{j,k=1}^n x_j y_k \, \int_0^1 \de t \int_0^1 \de s \, s \,
\theta_{s x + s t y}[B^{jk}](\omega) ,
$$
where we have used a parametrization of the flux involving the components of the magnetic field in the canonical basis of
$\X=\R^n$.
\end{proof}
By (\ref{chec}) one easily sees that $\left(\k^B\right)_\o=\k^{B_\o}$, where the l.h.s. was defined in
Proposition \ref{zisan}, while
$$
\k^{B_\o}(z;x,y):=\exp \bigl (-i\Gamma^{B_\o}\langle z,z+x,z+x+y\rangle \bigr ) .
$$

\subsection{Poisson algebras}\label{defaq}

We intend now to define a Poisson structure (cf. \cite{La,MR}) on spaces of functions that are smooth under the action
$\th\times\tau^*$ of $\Xi$ on $\O\times\X^*$. This Poisson algebras can be represented by families of subalgebras of
$BC^{\infty}(\Xi)$, indexed essentially by the orbits of $\O$, each one endowed with the Poisson structure induced by a magnetic
symplectic form \cite{MP2}. For simplicity, we shall concentrate on a  Poisson
subalgebra consisting of functions which have Schwartz-type
behavior in the variable $\xi\in\X^*$. For this smaller
algebra of functions, we will prove strict deformation quantization in section~\ref{defoq}.
One can also define $C^\infty(\O\times\X^*)$ in terms of the action $\th\otimes\tau^*$; this one is also a Poisson
algebra, but we will not need it here.

When necessary, we shall use $f(\xi)$ as short-hand notation for $f(\cdot,\xi)$,
\ie~$f(\o,\xi) = \bigl ( f(\xi) \bigr )(\o)$ for $(\o,\xi)\in\O\times\X^*$, and we will think of $f(\cdot,\xi)$ as an element
of some algebra of functions on $\O$. Note that
\begin{align*}
    BC^\infty(\O\times\X^*) = \bigl \{ f \in BC(\O\times\X^*) \mid f(\cdot,\xi) \in BC^\infty(\O)\\ {\rm and}\ f(\o,\cdot)\in
    BC^\infty(\X^*),\ \forall\o\in \O,\xi\in\X^* \bigr \} .
\end{align*}
\begin{definition}\label{spaciu}
We say that $f\in BC^{\infty}(\O \times \X^*)$ belongs to $\S \bigl (\X^*;C_0^{\infty}(\O) \bigr )$ if
\begin{enumerate}[(i)]
  \item $\partial^{\beta}_\xi f(\xi)\in C_0^\infty(\O),\ \forall\xi\in\X^*$ and
  \item $\snorm{f}_{a \alpha \beta} := \sup_{\xi \in \X^*} \bnorm{\xi^a \delta^{\alpha} \partial_{\xi}^{\beta}
  f(\xi)}_{C_0(\O)} < \infty$ for all $a , \alpha , \beta \in \N^n$.
\end{enumerate}
\end{definition}

\begin{proposition}\label{pesti}
We assume from now on that $B^{jk} \in BC^\infty(\O)$ for any $j,k=1,\dots,n$.
\begin{enumerate}[(i)]
\item
$BC^\infty(\O\times\X^*)$ is a Poisson algebra under
point-wise multiplication and the Poisson bracket
\begin{equation}\label{oisson}
\{f,g\}_B := \sum_{j=1}^n \bigl ( \partial_{\xi_j} f \, \delta_j g - \delta_j f \, \partial_{\xi_j} g \bigr ) -
\sum_{j,k}B^{jk} \, \partial_{\xi_j} f \, \partial_{\xi_k} g .
\end{equation}
\item
$\mathcal{S} \bigl ( \X^* ; C_0^{\infty}(\Omega) \bigr )$ is a Poisson subalgebra of $BC^{\infty}(\O\times\X^*)$.
\end{enumerate}
\end{proposition}

\begin{proof}
The two vector spaces are stable under point-wise multiplication and derivations with respect to $\xi$ and along orbits in
$\O$ via $\partial_\xi$ and $\delta$, respectively. They are also stable under multiplication with elements of
$BC^\infty(\O)$. The axioms of a Poisson algebra are verified by direct computation.
\end{proof}

To analyze the quantum calculus which is to be defined below, a change of realization is useful.
Defining $\mathcal{S} \bigl ( \X ; C_0^{\infty}(\Omega) \bigr )$ as in Definition \ref{spaciu}, but with $\X^*$ replaced with
$\X$, we transport by the partial Fourier transformation the Poisson structure
from $\mathcal{S} \bigl ( \X^* ; C_0^{\infty}(\Omega) \bigr )$ to $\mathcal{S} \bigl ( \X ; C_0^{\infty}(\Omega) \bigr )$
setting
\begin{equation}\label{bunica}
(\Phi \diamond_0 \Psi)(\omega ; x) := ( 1 \otimes \mathcal F )^{-1} \bigl ( ( 1 \otimes \mathcal F ) \Phi \cdot (
1 \otimes \mathcal F )\Psi \bigr )(\omega ; x)
= \int_{\X} \de y \, \Phi(\o ; y) \, \Psi(\o ; x-y)
\end{equation}
and
\begin{align}\label{bunika}
\{\Phi , \Psi\}^B :=& \; (1\otimes \mathcal F)^{-1} \bigl \{ (1\otimes \mathcal F)\Phi ,(1\otimes \mathcal F)\Psi \bigr
\}_B
\notag \\
=& \; - i \sum_{j=1}^n \bigl ( Q_j \Phi \diamond_0 \delta_j \Psi - \delta_j \Phi \diamond_0 Q_j \Psi \bigr ) +
\sum_{j,k=1}^n B^{jk} \, \bigl ( Q_j \Phi \diamond_0 Q_k\Psi \bigr ) ,
\end{align}
where $(Q_j \Phi)(x) = x_j \Phi(x)$ defines the multiplication operator by $x_j$. Obviously this also makes sense on larger
spaces.

\medskip
To get a better idea of the Poisson structure of $BC^{\infty}(\O \times \X^*)$, we will exploit the orbit structure
of the dynamical system $ \bigl ( \O \times \X^* , \theta \otimes \tau^* , \X \times \X^* \bigr )$ and relate this big
Poisson algebra to a family of smaller, symplectic-type ones.
For each $\o\in\O$, we can endow $\Xi = \X \times \X^*$ with a symplectic form
\begin{align*}
\bigl [ \sigma^B_{\o} \bigr ]_Z(X,Y) := y \cdot \xi - x \cdot \eta + B_\o(z)(x,y) = \sum_{j=1}^n \bigl ( y_j \, \xi_j - x_j
\, \eta_j ) + \sum_{j,k = 1}^n B^{jk} \bigl ( \theta_z[\omega] \bigr ) \, x_j \, y_k ,
\end{align*}
which makes the pair $\bigl ( \Xi , \sigma^B_{\omega})$ into a symplectic space. This canonically defines a Poisson bracket
\begin{equation}\label{scarz}
\{ f , g \}_{B_{\omega}} := \sum_{j = 1}^n \bigl ( \partial_{\xi_j} f \, \partial_{x_j} g - \partial_{x_j} f \,
\partial_{\xi_j} g \bigr ) - \sum_{j , k = 1}^n B_{\omega}^{jk} \, \partial_{\xi_j} f \, \partial_{\xi_k} g .
\end{equation}

\begin{proposition}\label{furnica}
\begin{enumerate}[(i)]
\item For each $\o \in \O$, the map
\begin{align*}
\pi_{\omega} := \theta_{\omega} \otimes 1:\bigl ( BC^{\infty}(\O\times\X^*) , \cdot, \{\cdot,\cdot\}_B
\bigr )\rightarrow\bigl ( BC^{\infty}(\Xi) , \cdot, \{\cdot,\cdot\}_{B_\o} \bigr )
 \end{align*}
 is a Poisson map, \ie \ for all $f , g \in BC^{\infty}(\Omega \times \X^*)$
\begin{align*}
\pi_{\omega} \bigl ( f \cdot g \bigr ) &= \pi_{\omega}(f) \cdot \pi_{\omega}(g),
\ \ \ \ \ \pi_{\omega} \bigl ( \{ f , g \}_B \bigr ) = \bigl \{ \pi_{\omega}(f) , \pi_{\omega}(g) \bigr \}_{B_{\omega}}\,.
\end{align*}

\item If $\o,\o'\in\O$ belong to the same orbit, the corresponding Poisson maps are connected by a symplectomorphism
(they may be called {\rm equivalent representations of the Poisson algebra}).
\end{enumerate}
\end{proposition}

\begin{proof}
 We use the notation
$f_{\omega} := \pi_{\omega}(f)$ for $f \in BC^{\infty}(\Omega \times \X^*)$ and $\omega \in \Omega$.

\begin{enumerate}[(i)]
\item For any $\omega \in \Omega$, $f , g \in BC^{\infty}(\Omega \times \X^*)$, we have
\begin{align*}
(f \, g)_{\omega}(x,\xi)= &\bigl ( (f \, g) \circ (\theta_{\omega} \otimes 1) \bigr )(x,\xi)
= f \bigl ( \theta_{\omega}(x) , \xi \bigr ) \, g \bigl ( \theta_{\omega}(x) , \xi \bigr )
= (f_{\omega} \, g_{\omega})(x,\xi).
\end{align*}
Similarly, $\bigl ( \bigl \{ f , g \bigr \}_B \bigr )_{\omega} = \bigl \{ f_{\omega} , g_{\omega} \bigr \}_{B_{\omega}}$
follows from direct computation, using
\begin{align*}
\partial_{x_j} f_{\omega} &= \partial_{x_j} \bigl ( f \circ (\theta_{\omega} \otimes 1) \bigr )
= (\delta_j f) \circ (\theta_{\omega} \otimes 1)
= (\delta_j f)_{\omega}
.
\end{align*}
\item If there exists $z\in\X$ such that
$\th_z[\o] = \o'$, then
$$
\th_{\o'}\otimes 1=\left(\th_\o\otimes 1\right)\circ\left(\tau_z\otimes 1\right),
$$
where $\tau_z\otimes 1:\left(\Xi,\sigma^B_\o\right)\rightarrow\left(\Xi,\sigma^B_{\o'}\right)$
is a symplectomorphism.
\end{enumerate}
\end{proof}

\begin{remark}
It is easy to see that the mapping
$$
\pi_{\omega} := \theta_{\omega} \otimes 1 : \S \bigl ( \X^*,C_0^\infty(\O) \bigr ) \longrightarrow \S \bigl (
\X^*,\A_\o^\infty \bigr )
$$
is a surjective morphism of Poisson algebras, for any $\o\in\O$. On the second space we consider the Poisson structure
defined by the magnetic field $B_\o$, as in \cite{MP2}.
\end{remark}

For any $\o\in\O$ we define the stabilizer $\X_\o:= \bigl \{x\in\X\mid\th_x[\o]=\o \bigr \}$. This is a closed subgroup
of $\X$, the same for all $\o$ belonging to a given orbit. We define the subspace of
$\O$ on which the action $\th$ is free:
$$
\O_0:= \bigl \{\o\in\O\mid \X_\o=\{0\} \bigr \}.
$$
Obviously $\O_0$ is invariant under $\th$ and $\O_0\times\X^*\,$ is invariant under the free action $\th\otimes\tau^*$,
so we can consider the Poisson algebra $BC^\infty(\O_0\times\X^*)$ with point-wise multiplication and Poisson
bracket~(\ref{oisson}).

For any $\mathcal O\in\mathbf O(\O_0)$ (the family of all the orbits of
the space $\O_0$) we choose a point $\o(\mathcal O)\in\mathcal O$. Then
$$
\th_{\o(\mathcal O)}\otimes 1:\X\times\X^*\longrightarrow\O_0\times\X^*
$$
is a continuous injection with range $\mathcal O\times\X^*$ (which is one of the orbits of $\O_0\times\X^*$ under the action
$\th\times\tau^*$). Of course, one has (disjoint union)
$$
\O_0\times\X^*=\bigsqcup_{\mathcal O\in\mathbf O(\O_0)} \mathcal{O} \times \X^*.
$$
In addition, $\th_{\o(\mathcal O)}\otimes 1$ is a Poisson mapping on $\Xi=\X\times\X^*$ if one considers the Poisson
structure induced by the symplectic form $\sigma^B_{\o(\mathcal O)}$.

Referring to Definition I.2.6.2 in \cite{La}, we notice that
$\O_0\times\X^*$ {\it is a Poisson space}.

\section{Quantum}\label{clll}

\subsection{Magnetic twisted crossed products}\label{calll}

\begin{definition}\label{turment}
We call {\rm twisted $C^*$-dynamical system} a quadruplet $(\A,\th,\k,\X)$, where $\th$ is an action of $\X=\R^n$ on the
(abelian) $C^*$-algebra $\A$ and $\k$ is a normalized $2$-cocycle on $\A$ with respect to $\th$.
\end{definition}
Starting from a twisted $C^*$-dynamical system, one can construct twisted crossed product
$C^*$-algebras \cite{PR1,PR2,MPR1} (see also references therein).
Let $L^1(\X; \A)$ be the complex vector space of $\A$-valued Bochner integrable functions on $\X$ and $L^1$-norm
$$
\norm{\Phi}_{L^1} :=\int_{\X} \de x \, \norm{\Phi(x)}_{\A} .
$$
For any $\Phi,\Psi \in L^1(\X;\A)$ and $x \in \X$, we define the product
\begin{equation*}
(\Phi \diamond^{\kappa} \Psi)(x):=\int_{\X}\de y\;\theta_{\frac{y-x}{2}}\!\left[\Phi(y)
\right]\;\!\theta_{\frac{y}{2}}\!\left[\Psi(x-y)\right]\;\!\theta_{-\frac{x}{2}}\!\left[\k(y,x-y)\right]
\end{equation*}
and the involution $\Phi^{\diamond^{\kappa}}(x):=\overline{\Phi(-x)}$. With these two operations, $\bigl ( L^1(\X;\A) ,
\diamond^{\kappa} , {}^{\diamond^{\kappa}} \bigr )$ forms a Banach-$\ast$-algebra.
\begin{definition}\label{primel}
The enveloping $C^*$-algebra of $L^1(\X;\A)$ is called {\rm the twisted crossed product} $\A\rtimes^\kappa_\th\X$.
\end{definition}
We are going to indicate now the relevant twisted crossed products, also introducing Planck's constant $\hb$ in the formalism.
We define
$$
\theta^{\hb}_x:=\theta_{\hb x}\ \ {\rm and}\ \ \k^{B,\hb}(x,y)=\k^{\frac{B}{\hb}}(\hb x,\hb y),
$$
which means
$$
\k^{B,\hb}(\o;x,y)=e^{-\frac{i}{\hb}\Gamma^{B_\o}\expval{0,\hb x,\hb x+\hb y}},\ \ \ \ \ \forall\,x,y\in\X,\ \o\in\O,
$$
and check easily that $\bigl ( C_0(\O),\theta^\hb,\kappa^{B,\hb},\X \bigr )$ is a twisted $C^*$-dynamical system for any
$\hb\in (0,1]$. It will be useful to introduce $\Lambda^B_{\hbar}(x,y)$ via
\begin{align*}
\theta_{- \frac{\hbar}{2} x}[\kappa^{B,\hbar}(\omega ; x,y)] = e^{- \frac{i}{\hbar} \Gamma^{B_{\omega}}
\expval{-\frac{\hbar}{2} x , \hbar y - \frac{\hbar}{2} x , \frac{\hbar}{2} x}} =:
e^{- i \hbar \Lambda^{B_{\omega}}_{\hbar}(x,y)},
\end{align*}
as short-hand notation for the phase factor.
This scaled magnetic flux can be parametrized explicitly as
\begin{align}\label{lam}
\Lambda^{B}_{\hbar}(x,y) = \sum_{j,k=1}^n y_j \, (x_k-y_k) \int_0^1 \de t \int_0 ^t \de s \, \theta_{\hbar
(s-\nicefrac{1}{2})x+\hbar(t-s)y}[B^{jk}].
\end{align}
Plugging this particular choice of $2$-cocycle and $\X$ action into the general form of the product, one gets
\begin{align*}
(\Phi \diamond^B_{\hbar} \Psi)(x) &= \int_{\X} \de y \, \theta_{\frac{\hbar}{2}(y-x)}[\Phi(y)] \,
\theta_{\frac{\hbar}{2} y}[\Psi(x-y)] \, e^{- i \hbar \Lambda^B_{\hbar}(x,y)}
.
\end{align*}

The twisted crossed product $C^*$-algebra $\A\rtimes_{\th^\hb}^{\k^{B,\hb}}\X$ will be denoted simply by
$\mathfrak{C}^{B}_\hb$ with self-adjoint part $\mathfrak{C}^{B}_{\hb,\R}$ and norm $\norm{\cdot}^B_\hb$. We also call
$\mathfrak{C}_0$ the enveloping $C^*$-algebra of $L^1(\X;\A)$ with the commutative product $\diamond_0$; it is
isomorphic with $C_0(\X^*;\A)\cong C(\X^*)\otimes\A$.

A quick computation shows that $\pi^\hb_{\omega}:=\th^\hb_\o\otimes 1$ intertwines the involutions associated to
the $C^*$-algebras $\mathfrak{C}^B_{\hbar}$
and $\mathcal{A}_{\omega} \rtimes^{\kappa^{B_{\omega} , \hbar}}_{\tau^{\hbar}} \X$, \ie~
$\pi^\hb_{\omega}(\Phi^{\diamond^B_{\hbar}}) = \pi^\hb_{\omega}(\Phi)^{\diamond^{B_{\omega}}_{\hbar}}$ is satisfied for every
$\Phi \in \mathfrak{C}^B_{\hbar}$. A slightly more cumbersome task is the verification of
$\pi^\hb_{\omega}(\Phi \diamond^B_{\hbar} \Psi) = \pi^\hb_{\omega}(\Phi) \diamond^{B_{\omega}}_{\hbar} \pi^\hb_{\omega}(\Psi)$.
For any $\Phi,\Psi\in L^1(\X;\A)$ and $z , x \in \X$, we have
$$\bigl [ \pi^\hb_{\omega}(\Phi \diamond^B_{\hbar} \Psi) \bigr ](z ; x) $$
$$=\int_{\X} \dd y \, \Bigl
( \theta_{\frac{\hbar}{2}(y-\cdot)}[\Phi(y)] \, \theta_{\frac{\hbar}{2} y}[\Psi(\cdot-y)] \, e^{- \frac{i}{\hbar}
\Gamma^{B} \expval{\cdot - \frac{\hbar}{2} \cdot , - \frac{\hbar}{2} \cdot + \hbar y , \frac{\hbar}{2} \cdot}} \Bigr )
\circ (\theta^\hbar_{\omega} \otimes 1)(z ; x) $$
$$= \int_{\X} \dd y \, \Phi \bigl ( \theta_{\hbar z + \frac{\hbar}{2}(y-x)}[\omega] , y \bigr ) \, \Psi \bigl
( \theta_{\hbar z + \frac{\hbar}{2} y}[\omega] , x-y \bigr ) \, e^{- \frac{i}{\hbar} \Gamma^{B_{\theta_{\hbar z}[\omega]}}
\expval{- \frac{\hbar}{2} x , - \frac{\hbar}{2} x + \hbar y , \frac{\hbar}{2} x}}$$
$$= \int_{\X} \dd y \,  \tau_{\frac{\hbar}{2}(y-x)}[\pi^\hbar_{\omega}(\Phi)(y)] (z) \,  \tau_{\frac{\hbar}{2}
y}[\pi^\hbar_{\omega}(\Psi)(x-y)] (z) \, e^{- \frac{i}{\hbar} \Gamma^{B_{\omega}} \expval{\hbar z - \frac{\hbar}{2} x ,
\hbar z - \frac{\hbar}{2} x + \hbar y , \hbar z + \frac{\hbar}{2} x}}$$
$$= \bigl [ \pi^\hb_{\omega}(\Phi) \diamond^{B_{\omega}}_{\hbar} \pi^\hb_{\omega}(\Psi) \bigr ](z ; x).$$

It follows easily that $\{ \pi^\hb_{\omega} \}_{\omega \in \Omega}$ defines by extension a family of epimorphisms
$$
\pi^\hbar_{\omega} : \mathfrak{C}^B_{\hbar} \longrightarrow \A_{\o} \rtimes_{\tau^\hb}^{\k^{B_\o,\hb}} \X
$$
that map a twisted crossed product defined in terms of $C_0(\Omega)$ onto more concrete $C^*$-algebras defined in terms of
subalgebras $\A_{\omega}$ of $BC_u(\X)$.

As we have seen, $\S \bigl (\X^*;C_0^{\infty}(\O) \bigr )$ is a Poisson subalgebra of $BC^{\infty}(\O\times\X^*)$.
For strict deformation quantization we also need that it is a *-subalgebra of each of the $C^*$-algebras $\mathfrak C^B_\hb$.
Since $\S \bigl (\X^*;C_0^{\infty}(\O) \bigr )$ is stable under involution, this will follow from

\begin{proposition}\label{intrucul}
If $B^{jk}\in BC^{\infty}(\O)$, then $\S\left(\X;C_0^\infty(\O)\right)$ is a subalgebra of\\
$\bigl ( L^1 \bigl ( \X ; C_0(\O) \bigr ), \diamond^B_{\hbar} \bigr )$.
\end{proposition}

\begin{proof}
Let $\Phi , \Psi \in \S \bigl ( \X ; C_0^{\infty}(\Omega) \bigr )$. As $\S \bigl ( \X ; C_0^{\infty}(\Omega) \bigr )$ is a
subspace of $L^1 \bigl ( \X ; C_0(\Omega) \bigr )$, $\Phi \diamond^B_{\hbar} \Psi$ exists in $L^1 \bigl ( \X ; C_0(\Omega)
\bigr )$. To prove the product $\Phi \diamond^B_{\hbar} \Psi$ is also in $\S \bigl ( \X ; C_0^{\infty}(\Omega) \bigr )$,
we need to estimate all semi-norms: let $a , \alpha , \beta \in \N^n$. First, we show that we can exchange differentiation
with respect to $x$ and along orbits with integration with respect to $y$ via Dominated Convergence, \ie~that for all $x$
and $\omega$
$$\bigl ( x^a \partial_x^{\alpha} \delta^{\beta} (\Phi \diamond^B_{\hbar} \Psi) \bigr )(\omega ; x)$$
$$=\int_{\X} \de y \, x^a
\partial_x^{\alpha} \delta^{\beta} \Bigl ( \Phi \bigl ( \theta_{\frac{\hbar}{2}(y-x)}[\omega] , y \bigr ) \, \Psi \bigl
(\theta_{\frac{\hbar}{2} y}[\omega] , x-y \bigr ) \, e^{- i \hbar \Lambda^{B_{\omega}}_{\hbar}(x,y)} \Bigr )
=: \int_{\X} \de y \, I^a_{\alpha \beta}(\omega ; x , y)$$
holds. Hence, we need to estimate the absolute value of $I^a_{\alpha \beta}$ uniformly in $x$ and $\omega$ by an integrable
function. To do that, we write out the derivatives involved in $I^a_{\alpha \beta}$,
$$I^a_{\alpha \beta}(x,y)= x^a \partial_x^{\alpha} \delta^{\beta} \Bigl ( \theta_{\frac{\hbar}{2}(y-x)}[\Phi(y)] \,
\theta_{\frac{\hbar}{2} y}[\Psi(x-y)] \, e^{- i \hbar \Lambda^B_{\hbar}(x,y)} \Bigr )$$
$$= x^a \!\sum_{\substack{\alpha' \!+\! \alpha''\! +\! \alpha'''\! = \!\alpha \\ \beta'\! + \!\beta''\! +\! \beta'''\! =
\! \beta}}\!\left(\frac{-\hbar}{2}\right)^{\sabs{\alpha'}}\theta_{\frac{\hbar}{2}(y-x)} \bigl [ \delta^{\alpha' + \beta'}
\Phi(y) \bigr ]\theta_{\frac{\hbar}{2} y} \bigl [ \partial_x^{\alpha''} \delta^{\beta''} \Psi(x - y) \bigr ]
\partial_x^{\alpha'''} \delta^{\beta'''} e^{- i \hbar \Lambda^B_{\hbar}(x,y)}.$$
Taking the $C_0(\Omega)$ norm of the above expression and using the triangle inequality, $\hbar \leq 1$, the fact that
$\th_z$ is an isometry as well as the estimates on the exponential of the magnetic flux from
Lemma~\ref{mag_twistedXprod:lem:flux_estimates}~(ii), we get
$$\bnorm{I^a_{\alpha \beta}(x,y)}_{\A}
\leq \sabs{x^a} \sum_{\substack{\alpha'\! +\! \alpha''\! +\! \alpha'''\! =\! \alpha \\
\beta' \!+ \!\beta''\! +\! \beta'''\! =\! \beta}}\!\! \left(\frac{\hbar}{2}\right)^{\sabs{\alpha'}}\bnorm{\theta_{\frac{\hbar}{2}(y-x)}
\bigl [ \delta^{\alpha' + \beta'} \Phi(y) \bigr ]}_{\A} \bnorm{\theta_{\frac{\hbar}{2} y} \bigl [ \partial_x^{\alpha''}
\delta^{\beta''} \Psi(x - y) \bigr ]}_{\A}\cdot$$
$$\cdot\bnorm{\partial_x^{\alpha'''} \delta^{\beta'''}
e^{- i \hbar \Lambda^B_{\hbar}(x,y)}}_{\A}
\leq \left ( \prod_{j = 1}^n \bigl ( \sabs{y_j} + \sabs{x_j - y_j} \bigr )^{a_j} \right ) \cdot$$
$$\cdot\sum_{\substack{\alpha' +
\alpha'' + \alpha''' = \alpha \\ \beta' + \beta'' + \beta''' = \beta}} \bnorm{\delta^{\alpha' + \beta'} \Phi(y)}_{\A} \,
\bnorm{\partial_x^{\alpha''} \delta^{\beta''} \Psi(x - y)}_{\A}
 \cdot
\sum_{\abs{b} + \abs{c} = 2 (\sabs{\alpha'''} + \sabs{\beta'''})} K_{bc} \, \sabs{y^b} \, \sabs{(x - y)^c}
$$
$$= \sum_{\substack{\alpha' + \alpha'' + \alpha''' = \alpha \\ \beta' + \beta'' + \beta''' = \beta}} \sum_{\abs{b} + \abs{c}
\leq \abs{a} + 2 \sabs{\alpha'''} + 2 \sabs{\beta'''}} \tilde{K}_{bc} \bnorm{(Q^b \delta^{\alpha' + \beta'} \Phi)(y)}_{\A}
\bnorm{(Q^c \partial_x^{\alpha''} \delta^{\beta''} \Psi)(x - y)}_{\A}.$$
The polynomial with coefficients $\tilde{K}_{bc}$ comes from multiplying the other two polynomials in the $\sabs{y_j}$ and
$\sabs{x_j - y_j}$. Taking the supremum in $x$ only yields a function in $y$ (independent of $x$ and $\omega$) which is
integrable and dominates $\babs{I^a_{\alpha \beta}(\omega ; x,y)}$ since the right-hand side is a finite sum of Schwartz
functions in $y$,
\begin{align*}
\sup_{x \in \X} \bnorm{I^a_{\alpha \beta}(x,y)}_{\A} &\leq \sum_{\substack{\alpha' + \alpha'' + \alpha''' = \alpha \\
\beta' + \beta'' + \beta''' = \beta \\ \abs{b} + \abs{c} \leq \abs{a} + 2 \sabs{\alpha'''} + 2 \sabs{\beta'''}}}
\tilde{K}_{bc} \,  \bnorm{(Q^b \delta^{\alpha' + \beta'} \Phi)(y)}_{\A} \, \bnorm{Q^c \partial_x^{\alpha''}
\delta^{\beta''} \Psi}_{000}
\\
&= \sum_{\substack{\alpha' + \alpha'' + \alpha''' = \alpha \\ \beta' + \beta'' + \beta''' = \beta \\ \abs{b} + \abs{c}
\leq \abs{a} + 2 \sabs{\alpha'''} + 2 \sabs{\beta'''}}} \tilde{K}_{bc} \,  \bnorm{(Q^b \delta^{\alpha' + \beta'}
\Phi)(y)}_{\A} \, \bnorm{\Psi}_{c \alpha'' \beta''}
.
\end{align*}
Hence, by Dominated Convergence, it is permissible to interchange differentiation and integration. To estimate the semi-norm
of the product, we write for an integer $N$ such that $2N\ge n+1$
\begin{align}
\bnorm{\Phi \diamond^B_{\hbar} \Psi}_{a \alpha \beta} = \sup_{\substack{x \in \X \\ \omega \in \Omega}} \abs{\int_{\X}
\de y \, I^a_{\alpha \beta}(\omega ; x , y)}
\leq \int_{\X}\frac{\de y}{\left<y\right>^{2N}} \, \left<y\right>^{2N} \sup_{x \in \X} \bnorm{I^a_{\alpha
\beta}(x , y)}_{\A}
\notag \\
\leq C_1(N) \, \sup_{x,y \in \X}\left(\left<y\right>^{2N} \bnorm{I^a_{\alpha \beta}(x , y)}_{\A}\right) \le C_2(N)
\, \max_{\abs{b} \le 2N} \sup_{x,y \in \X} \bnorm{y^b I^a_{\alpha \beta}(x , y)}_{\A}.
\label{mag_twistedXprod:eqn:proof_norm_estimate}
\end{align}
The right-hand side involves semi-norms associated to $\S \bigl ( \X \times \X ; C^{\infty}_{0}(\Omega) \bigr )$ which we will
estimate in terms of the semi-norms of $\Phi$ and $\Psi$, by arguments similar to those leading to the domination of
$\bnorm{I^a_{\alpha \beta}(x,y)}_{\A}$.

Thus, we have estimated $\bnorm{\Phi \diamond^B_{\hbar} \Psi}_{a \alpha \beta}$ from above by a finite number of semi-norms
of $\Phi$ and $\Psi$ and $\Phi \diamond^B_{\hbar} \Psi \in \S \bigl ( \X ; C^{\infty}(\Omega) \bigr )$.
\end{proof}

\subsection{Twisted symbolic calculus}\label{tvity}

It is useful to transport the composition law $\diamond^B_\hb$ by partial Fourier transform
$
1\otimes\mathcal F:\S(\X;C_0^\infty(\O)) \longrightarrow \S(\X^*;C_0^\infty(\O)),
$
setting
\begin{equation}\label{novela}
f\sharp^B_\hb g:=(1\otimes\mathcal F)\left[(1\otimes\mathcal F)^{-1}f\,\diamond^B_\hb\,(1\otimes\mathcal F)^{-1}g\right].
\end{equation}
In this way one gets a multiplication on $\S(\X^*;C_0^\infty(\O))$ which generalizes the magnetic Weyl composition of symbols
of \cite{MP1,MP2,IMP1} (and to which it reduces, actually, if $\Omega$ is just a compactification of the configuration space
$\X$). Together with complex conjugation, they endow $\S(\X^*;C_0^\infty(\O))$ with the structure of a $^*$-algebra.
After a short computation one gets
$$(f \sharp^B_{\hbar} g)(\omega,\xi)$$
$$= (\pi \hbar)^{-2n} \int_{\X} \de y \int_{\X^*} \de \eta \int_{\X} \de z \int_{\X^*} \de
\zeta \, e^{i \frac{2}{\hbar}(z \cdot \eta - y \cdot \zeta)} \, e^{- \frac{i}{\hbar} \Gamma^{B_{\omega}} \expval{\hbar y -
\hbar z , \hbar y + \hbar z , \hbar z - \hbar y}}\cdot $$
$$ \cdot f \bigl ( \theta_y[\omega] , \xi + \eta \bigr ) \, g \bigl ( \theta_z[\omega] , \xi + \zeta \bigr )$$
$$= (\pi \hbar)^{-2n} \int_{\X} \de y \int_{\X^*} \de \eta \int_{\X} \de z \int_{\X^*} \de \zeta \,
e^{i \frac{2}{\hbar}\sigma[(y,\eta),(z,\zeta)]} \, e^{- \frac{i}{\hbar} \Gamma^{B_{\omega}}
\expval{\hbar y - \hbar z , \hbar y + \hbar z ,\hbar z - \hbar y}}
\cdot $$
\begin{equation}\label{chiueisan}
\cdot\bigl ( \Theta_{(y,\eta)}[f] \bigr )(\omega,x) \, \bigl ( \Theta_{(z,\zeta)}[g] \bigr )(\omega,x),
\end{equation}
where $\sigma[(y,\eta),(z,\zeta)]:=z\cdot\eta-y\cdot\zeta$ is the canonical symplectic form on $\Xi:=\X\times\X^*$ and
\begin{align*}
\bigl ( \Theta_{(y,\eta)}[f] \bigr )(\omega , \xi) \equiv \bigl ( (\theta_y \otimes \tau^*_{\eta}) [f] \bigr ) (\omega , \xi)
= f \bigl ( \theta_y[\omega] , \xi + \eta \bigr ).
\end{align*}
This formula should be compared with the product giving Rieffel's quantization \cite{Rie1}.

We note that $1\otimes\mathcal F$ can be extended to $L^1(\X;C_0(\O))$ and then to $\mathfrak C^B_\hb$. So we get a
$C^*$-algebra $\mathfrak B^B_\hb$, isomorphic to $\mathfrak C^B_\hb$, on which the product is an extension of the
twisted composition law~(\ref{chiueisan}). From the bijectivity of the partial Fourier transform and Proposition~\ref{intrucul}
we get the following

\begin{corollary}
If the components of the magnetic field $B$ are of class $BC^{\infty}(\Omega)$, then $\mathcal{S} \bigl
( \X^* ; C_0^{\infty}(\Omega) \bigr )$ is a Fréchet $^*$-subalgebra of $\mathfrak{B}^B_{\hbar}$.
\end{corollary}

\subsection{Representations}\label{preze}

We first recall the definition of covariant representations of a magnetic $C^*$-dynamical system and the way they are used to
construct representations of the corresponding $C^*$-algebras. We denote by $\U(\H)$ the group of unitary operators in
the Hilbert space $\H$ and by $\mathbf B(\H)$ the $C^*$-algebra of all the linear bounded operators on $\H$.

\begin{definition}\label{RCT}
Given a magnetic $C^*$-dynamical system $(\A,\theta^\hb,\k^{B,\hb},\X)$, we call {\rm covariant representation} $(\H,r,T)$ a
Hilbert space $\H$ together with two maps $r:\A\rightarrow \mathbf B(\H)$ and $T:\X\rightarrow \U(\H)$ satisfying
\begin{enumerate}[(i)]
\item $r$ is a non-degenerate representation,
\item $T$ is strongly continuous and $\
T(x)\;\!T(y)=r[\k^{B,\hb}(x,y)]\;\!T(x+y),\quad \forall x,y\in \X$,
\item $T(x)\;\!r(\varphi)\;\!T(x)^*=r[\theta_x^\hb(\varphi)], \quad
\forall x\in\X,\;\varphi\in\A$.
\end{enumerate}
\end{definition}

\begin{lemma} If $(\H,r,T)$ is a covariant representation of $(\A,\theta^\hb,\k^{B,\hb},\X)$, then $\Rep^T_r$ defined on
$L^1(\X;\A)$ by
\begin{equation*}
\Rep_r^T(\Phi):=\int_\X \de y\,r\left[\theta^\hb_{y/2}\big(\Phi(y)\big)\right]T(y)
\end{equation*}
extends to a representation of $\mathfrak C^B_\hb$.
\end{lemma}

By composing with the partial Fourier transformation, one gets representations of the pseudodifferential $C^*$-algebra
$\mathfrak B^B_\hb$, denoted by
\begin{equation}\label{gigi}
\Op^T_r:\mathfrak B^B_\hb\rightarrow\mathbf B(\H),\ \ \ \ \ \Op^T_r(f):=\Rep^T_r\left[(1\otimes\mathcal F)^{-1}(f)\right].
\end{equation}

Given any $\o\in \Omega$, we shall now construct a  representation of $\mathfrak C^B_\hb$ in $\H=L^2(\X)$.
Let $r_\o$ be the representation of $\A$ in $\mathbf B(\H)$ given for  $\varphi\in\A$, $u \in \H$ and $x \in \X$ by
\begin{equation*}
[r_\o(\varphi)u](x) = [\theta_x(\varphi)](\o)\;u(x) \equiv \varphi\big(\theta_x[\o]\big)\;u(x).
\end{equation*}
Let also $T_\o^\hb$ be the map from $\X$ into the set of unitary operators on $\H$ given by
\begin{equation*}
[T^\hb_\o(y)u](x):=\k^{B,\hb}(\o;x/\hb,y)\;u(x+\hb y)=e^{-\frac{i}{\hb}\Gamma^{B_\o}\left<0,x,x+\hb y\right>}u(x+\hb y).
\end{equation*}

\begin{proposition}\label{banala}
$(\H,r_\o,T^\hb_\o)$ is a covariant representation of the magnetic twisted $C^*$-dynamical system.
\end{proposition}

\begin{proof}
 Use the definitions, Stokes Theorem for the magnetic field $B_\o$ and the identities
$$
\Gamma^{B_\o}\left<x,x+\hb y,x+\hb y+\hb z\right>=\Gamma^{B_{\th_x[\o]}}\left<0,\hb y,\hb y+\hb z\right>
$$
and
$$
\Gamma^{B_\o}\left<0,x+\hb y,x\right>=-\Gamma^{B_\o}\left<0,x,x+\hb y\right>,
$$
valid for all $x,y,z\in\X$ and $\o\in\O$.
\end{proof}

The integrated form $\Rep^\hb_\o:=\Rep_{r_\o}^{T^\hb_\o}$ has the following action on $L^1(\X;\A)$:
\begin{eqnarray}\label{eqsurrep}
\nonumber \big[\Rep^\hb_\o(\Phi)u\big](x)&=&\int_\X \de z \; \Phi\left(\theta_{x+\frac{\hb z}{2}}[\o];z\right)\;
\k^{B,\hb}(\o;x/\hb,z)\;u(x+\hb z) \\
&=& \hb^{-n}\int_\X \de y \; \Phi\left(\theta_{\frac{x+y}{2}}[\o];\frac{1}{\hb}(y-x)\right)
\;e^{-\frac{i}{\hb}\Gamma^{B_\o}\left<0,x,y\right>}\;u(y),
\end{eqnarray}
and the corresponding representation $\Op_\o^\hb$ of the $C^*$-algebra $\mathfrak B^B_\hb$ has the following form on suitable
$f \in \mathfrak B^B_\hb$:
\begin{equation}\label{espertu}
\left[\Op^\hb_\o(f)u\right](x)=(2\pi\hb)^{-n}\int_\X \de y\int_{\X^*}\de \xi \;e^{\frac{i}{\hb}(x-y)\cdot\xi}
f\left(\theta_{\frac{x+y}{2}}[\o];\xi\right)\;e^{-\frac{i}{\hb}\Gamma^{B_\o}\left<0,x,y\right>}\;u(y).
\end{equation}
It is clear that $\Op^\hb_\o$ is {\it not} a faithful representation, since (\ref{espertu}) only involves the values taken by
$f$ on $\mathcal O_\o\times\X^*$, where $\mathcal O_\o$ is the orbit passing through $\o$.
It is rather easy to show that the kernel of $\Op^\hb_\o$ can be identified with the twisted crossed product
$C_0(\mathcal Q_\o)\rtimes_{\th^\hb}^{\kappa^B,\hb}\X$, constructed as explained above, with $\O$ replaced by
$\mathcal Q_\o:=\overline{\mathcal O_\o}$, the quasi-orbit generated by the point $\o$.
\begin{remark}
The expert in the theory of quantum magnetic fields might recognize in (\ref{espertu}) the expression of a magnetic
pseudodifferential operator with symbol $f\circ(\th_\o\otimes 1)$, written in the transverse gauge
for the magnetic field $B_\o$. Then it will be a simple exercise to write down analogous representations associated to
continuous (fields of) vector potentials $A:\O\rightarrow\wedge^1\X$ generating the magnetic field (i.e. $B_\o=dA_\o$,
 $\forall\,\o\in\O$) and to check an obvious principle of gauge-covariance.
\end{remark}
We show now that the family of representations $\{\Op^\hb_\o\mid\o\in\O\}$ actually has as a natural index set the orbit
space of the dynamical system, up to unitary equivalence.

\begin{proposition}\label{surRep}
Let $\o,\o'$ be two elements of $\O$, belonging to the same orbit under the action $\th$.
Then, for any $\hb\in(0,1]$, one has $\Rep_\o^\hb \cong \Rep_{\o'}^\hb$ and $\Op_\o^\hb \cong \Op_{\o'}^\hb$
(unitary equivalence of representations).
\end{proposition}

\begin{proof}
By assumption, there exists an element $x_0$ of $\X$ such that $\theta_{x_0}[\o'] = \o$. For $u\in \H$
and $x\in \X$ we define the unitary operator
\begin{equation*}
\left(U_{\o,\o'}^\hb\;\!u\right)(x):= e^{-\frac{i}{\hb}\Gamma^{B_{\o'}}\left<0,x_0,x_0+x\right>}\;\!u(x+x_0)\ .
\end{equation*}
To show unitary equivalence of the two representations, it is enough to show that for all $\varphi\in\A$ and $y\in\X$
\begin{equation*}
U_{\o,\o'}^\hb\,r_{\o'}(\varphi)=r_\o(\varphi)\,U_{\o,\o'}^\hb\quad \hbox{ and }
\quad U_{\o,\o'}^\hb\,T_{\o'}^\hb(y) =T_\o^\hb(y)\,U_{\o,\o'}^\hb.
\end{equation*}
The first one is obvious. The second one reduces to
\begin{align*}
\Gamma^{B_{\o'}}\left<0,x_0,x_0+x\right> &+ \Gamma^{B_{\o'}}\left<0,x_0+x,x_0+x+\hb y\right>
= \\
&
= \Gamma^{B_{\o'}}\left<x_0,x_0+x,x_0+x+\hb y\right>+\Gamma^{B_{\o'}}\left<0,x_0,x_0+x+\hb y\right>,
\end{align*}
which is true by Stokes Theorem.
\end{proof}

\begin{remark}
The Proposition reveals what we consider to be the main practical interest of the formalism we develop
in the present article. To a fixed real symbol $f$ and to a fixed value $\hb$ of Planck's constant one associates a family
$\{H^\hb_\o:=\Op^\hb_\o(f)\mid\o\in\O\}$
of self-adjoint magnetic pseudodifferential operators on the Hilbert space $\H:=L^2(\X)$, indexed by the points of a
dynamical system $(\O,\th,\X)$ and satisfying the equivariance condition
\begin{equation}\label{echivar}
H_{\theta_x[\omega]}^\hb=\left(U_{\o,\theta_x[\omega]}^\hb\right)^{-1}H_{\o}^\hb\,U_{\o,\theta_x[\omega]}^\hb,\ \ \ \ \
\forall\,(\o,x)\in\O\times\X.
\end{equation}
In concrete situations, such equivariance conditions usually carry some physical meaning.
In a future publication we are going to extend the formalism to unbounded symbols $f$, getting
realistic magnetic Quantum Hamiltonians organized in equivariant families, which will be studied in
the framework of spectral theory.
\end{remark}
To define other types of representations, we consider now $\O$ endowed with a $\th$-invariant measure $\mu$.
Such measures always exist, since $\X$ is
abelian hence amenable. We set $\H'$ for the Hilbert space $L^2(\O,\mu)$ and consider first the faithful
representation: $\tilde{r}:\A \to \mathbf B(\H')$ with $[\tilde{r}(\varphi)v](\o):=\varphi(\o)\;\!v(\o)$ for all $v\in \H'$
and $\o \in \O$. Then, (by a standard construction in the theory of twisted crossed products)
the \emph{regular representation of the magnetic $C^*$-dynamical system $(\A,\theta^\hb,\k^{B,\hb},\X)$
induced by} $\tilde{r}$ is the covariant representation $\big(L^2(\X;\H'),r,T^\hb\big)$:
\begin{equation*}
r:\A\to \mathbf B\big[L^2 (\X;\H')\big], \quad
[r(\varphi)w](\omega;x):=\left(\tilde r(\th_x(\varphi))[w(x)]\right)(\o)
=\varphi\big(\theta_x(\o)\big)\;\!w(\o;x)\ ,
\end{equation*}
\begin{equation*}
T^\hb:\X\to\U [L^2(\X;\H')], \quad [T^\hb(y) w](\o;x):=\k^{B,\hb}(\o;x/\hb,y)\;\!w(\o;x+\hb y).
\end{equation*}
We identify freely $L^2(\X;\H')$ with $L^2(\Omega\times\X)$ with the obvious product measure, so $r(\varphi)$ is the operator
of multiplication by $\varphi\circ\th$ in $L^2(\Omega\times\X)$. Due to Stokes' Theorem, this is again a covariant
representation. The integrated form $\mathfrak{REP}^\hb:=\Rep^{T^\hb}_r$ associated to $(r,T^\hb)$ is given on $L^1(\X;\A)$ by
\begin{equation*}
\big[\mathfrak{REP}^\hb(\Phi)w\big](\o;x)=\hb^{-n}\int_\X \de y \; \Phi\left(\theta_\frac{x+y}{2}[\o];\frac{y-x}{\hb}\right)
\; e^{-\frac{i}{\hb}\Gamma^{B_\o}\left<0,x,y\right>}\;w(\o;y)
\end{equation*}
and it admits the direct integral decomposition
\begin{equation}\label{directint}
\mathfrak{REP}^\hb(\Phi)=\int^{\oplus}_{\O} \de \mu (\o)\;\!\Rep_\o^\hb(\Phi).
\end{equation}
The group $\X$, being abelian, is amenable, and thus the regular representation $\mathfrak{REP}^\hb$ is faithful.
The corresponding representation $\mathfrak{OP}^\hb: \mathfrak B^B_\hb \to \mathbf B\left[L^2(\X;\H')\right]$
is given for $f$ with partial Fourier transform in $L^1(\X;\A)$ by
\begin{equation*}
\big[\mathfrak{OP}^\hb(f)w\big](\o;x)=(2\pi\hb)^{-n}\int_\X \int_{\X^*} \de y  \de \eta  e^{\frac{i}{\hb}(x-y)\cdot \eta}
f\left(\theta_\frac{x+y}{2}[\o],\eta\right) e^{-\frac{i}{\hb}\Gamma^{B_\o}\left<0,x,y\right>} w(\o;y)\ .
\end{equation*}

\section{Asymptotic expansion of the product}\label{asymptotic_expansion}

The proof of strict deformation quantization hinges on the following Theorem:
\begin{theorem}[Asymptotic expansion of the product]\label{theorem:asymptotic_expansion_product}
Assume the components of $B$ are in $BC^{\infty}(\O)$. Let $\Phi , \Psi \in \S\left(\X;C_0^\infty(\O)\right)$
and $\hbar \in (0,1]$. Then the product $\Phi \tcB \Psi$ can be expanded in powers of $\hbar$,
\begin{equation}
\Phi \tcB \Psi = \Phi \tcz \Psi - \hbar \tfrac{i}{2} \{ \Phi , \Psi \}^B + \hbar^2 R^{\diamond,2}_\hb(\Phi,\Psi),
\end{equation}
where $\{ \Phi , \Psi \}^B$ is defined as in equation (\ref{bunika}). All terms are in $\S\left(\X;C_0^\infty(\O)\right)$
and $R^{\diamond,2}_\hb(\Phi,\Psi)$ is bounded uniformly in $\hbar$,
 $\parallel R^{\diamond,2}_\hb(\Phi,\Psi)\parallel_\hb^B \leq C$.
\end{theorem}
\begin{proof}
We are going to use  Einstein's summation convention, \ie~repeated indices in a product are summed over.
Two types of terms in the product formula need to be expanded in $\hbar$, the group action of $\mathcal{X}$ on $\Omega$,
$$\bigl ( \theta_{\frac{\hbar}{2} y}[\Phi(x)] \bigr )(\omega) = \Phi \bigl ( \theta_{\frac{\hbar}{2} y}[\omega] ; x)
= \Phi(\omega;x) + \hbar \int_0^1 \dd \tau \, \tfrac{1}{2} y_j \, \theta_{\tau \frac{\hbar}{2} y} \bigl[(\delta_j \Phi)(\omega ; x) \bigr ]$$
$$=: \Phi(\omega;x) + \hbar \, \bigl ( R^{\theta,1}_{\hb,y}(\Phi) \bigr )(\omega;x)$$
$$= \Phi(\omega;x) + \tfrac{\hbar}{2} y_j \, (\delta_j \Phi)(\omega ; x) + \hbar^2 \int_0^1 \dd \tau \,
\tfrac{1}{4} (1 - \tau) \, y_j y_k \, \theta_{\tau \frac{\hbar}{2} y} \bigl [ (\delta_j \delta_k \Phi)(\omega ; x )\bigr ]$$
$$=: \Phi(\omega;x) + \tfrac{\hbar}{2} y_j \, (\delta_j \Phi)(\omega ; x) + \hbar^2 \bigl ( R^{\theta,2}_{\hb,y}(\Phi)\bigr )(\omega ; x),$$
and the exponential of the magnetic flux,
$$e^{- i \hbar \Lambda^B_{\hbar}(x,y)} = 1 + \hbar \int_0^1 \dd \tau \, \frac{\dd}{\dd \epsilon} \bigl
( e^{- i \epsilon \Lambda^B_{\epsilon}(x,y)} \bigr ) \bigg \vert_{\epsilon = \tau \hbar}
=: 1 + \hbar \, R^{\kappa,1}_\hb(x,y)$$
$$= 1 - \hbar i \Lambda^B_0(x,y) + \hbar^2 \int_0^1 \dd \tau \, (1 - \tau) \, \frac{\dd^2}{\dd \epsilon^2}
\bigl ( e^{- i \epsilon \Lambda^B_{\epsilon}(x,y)} \bigr ) \bigg \vert_{\epsilon = \tau \hbar}$$
$$=: 1 - \hbar \tfrac{i}{2} B^{jk} \, y_j \, (x_k - y_k) + \hbar^2 \, R^{\kappa,2}_\hb(x,y).$$
We will successively plug these expansions into the product formula, keeping only terms of $\order(\hbar^2)$:
$$\bigl ( \Phi \tcB \Psi \bigr )(x)= \int_{\mathcal{X}} \dd y \, \Bigl (
\Phi(y) + \tfrac{\hbar}{2} (y_j - x_j) \, (\delta_j \Phi)(y) + \hbar^2 \bigl ( R^{\theta,2}_{\hb,y - x}(\Phi)
\bigr )(y) \Bigr ) \, \theta_{\frac{\hbar}{2} y} [\Psi(x-y)] \, e^{- i \hbar \Lambda^B_{\hbar}(x,y)}$$
$$= \int_{\mathcal{X}} \dd y \, \Phi(y) \, \Bigl (
\Psi(x-y) + \tfrac{\hbar}{2} y_j \, (\delta_j \Psi)(x-y) + \hbar^2 \bigl ( R^{\theta,2}_{\hb,y}(\Psi) \bigr )(x-y)
\Bigr ) \, e^{- i \hbar \Lambda^B_{\hbar}(x,y)}+ $$
$$+ \frac{\hbar}{2} \int_{\mathcal{X}} \dd y \, (y_j - x_j) \, (\delta_j \Phi)(y) \, \Bigl ( \Psi(x-y) + \hbar \bigl
( R^{\theta,1}_{\hb,y} (\Psi) \bigr )(x-y) \Bigr ) \, e^{- i \hbar \Lambda^B_{\hbar}(x,y)}+ $$
$$+ \hbar^2 \int_{\mathcal{X}} \dd y \, \bigl ( R^{\theta,2}_{\hb,y - x}(\Phi) \bigr )(y) \, \theta_{\frac{\hbar}{2} y}
[\Psi(x-y)] \, e^{- i \hbar \Lambda^B_{\hbar}(x,y)}$$
$$= \int_{\mathcal{X}} \dd y \, \Phi(y) \, \Psi(x-y) \, \Bigl ( 1 - \hbar i \Lambda^B_0(x,y) +
\hbar^2 R^{\kappa,2}_\hb(x,y) \Bigr )+$$
$$+ \frac{\hbar}{2} \int_{\mathcal{X}} \dd y \, \Bigl ( -  (\delta_j \Phi)(y) \, (Q_j \Psi)(x-y) +
(Q_j \Phi)(y) \, (\delta_j \Psi)(x-y) \Bigr ) \, \Bigl ( 1 + \hbar R^{\kappa,1}_\hb(x,y) \Bigr )+ $$
$$+ \hbar^2 \int_{\mathcal{X}} \dd y \, \Bigl [ \bigl ( R^{\theta,2}_{\hb,y - x}(\Phi) \bigr )(y) \,
\theta_{\frac{\hbar}{2} y} [\Psi(x-y)] \, e^{- i \hbar \Lambda^B_{\hbar}(x,y)}
+ \Phi(y) \, \bigl ( R^{\theta,2}_{\hb,y}(\Psi) \bigr )(x-y) \, e^{- i \hbar \Lambda^B_{\hbar}(x,y)}
- \Bigr . $$
$$\Bigl .
- \tfrac{1}{2} (\delta_j \Phi)(y) \, \bigl ( Q_j R^{\theta,1}_{\hb,y} (\Psi) \bigr )(x-y) \, e^{- i \hbar
\Lambda^B_{\hbar}(x,y)}\Bigr ]$$
$$= \int_{\mathcal{X}} \dd y \, \Phi(y) \, \Psi(x-y)+$$
$$+ \frac{\hbar}{2} \int_{\mathcal{X}} \dd y \,  \Bigl ( (Q_j \Phi)(y) \, (\delta_j \Psi)(x-y)
- (\delta_j \Phi)(y) \, (Q_j \Psi)(x-y)
- i B^{jk} \, (Q_j \Phi)(y) \, (Q_k \Psi)(x-y)
    \Bigr )+$$
$$+ \hbar^2 \int_{\mathcal{X}} \dd y \, \Bigl [ \Bigl ( \bigl ( R^{\theta,2}_{\hb,y - x}(\Phi) \bigr )(y) \,
\theta_{\frac{\hbar}{2} y} [\Psi(x-y)]
+ \Phi(y) \, \bigl ( R^{\theta,2}_{\hb,y}(\Psi) \bigr )(x-y) \Bigr ) \, e^{- i \hbar \Lambda^B_{\hbar}(x,y)}
- \Bigr . $$
$$ \Bigl .
- \tfrac{1}{2} (\delta_j \Phi)(y) \, \bigl ( Q_j R^{\theta,1}_{\hb,y} (\Psi) \bigr )(x-y) \, e^{- i \hbar
\Lambda^B_{\hbar}(x,y)}
+ \Bigr . $$
$$ \Bigl .
+ \tfrac{1}{2} \bigl ( (Q_j \Phi)(y) \, (\delta_j \Psi)(x-y)
- (\delta_j \Phi)(y) \, (Q_j \Psi)(x-y) \bigr ) \, R^{\kappa,1}_\hb(x,y)
+ \Bigr . $$
$$ \Bigl .
+ \Phi(y) \, \Psi(x-y) \, R^{\kappa,2}_\hb(x,y)
\Bigr ]$$
$$=: \bigl ( \Phi \tcz \Psi \bigr )(x) - \hbar \tfrac{i}{2} \{ \Phi , \Psi \}^B(x) + \hbar^2 \bigl
( R^{\diamond,2}_{\hbar}(\Phi,\Psi) \bigr )(x).$$

In the above, we have used $(y_j - x_j) \, \Psi(x-y) = - (Q_j \Psi)(x - y)$, $y_j \, \Phi(y) = (Q_j \Phi)(y)$ and
the explicit expression for $\Lambda^B_0(x,y)$.
Clearly, the leading-order and sub-leading-order terms are again in $\S\left(\X;C_0^\infty(\O)\right)$. Thus also
$R^{\diamond,2}_{\hbar}(\Phi,\Psi) = \hbar^{-2} \bigl ( \Phi \tcB \Psi - \Phi \tcz \Psi + \hbar \tfrac{i}{2} \{
\Phi , \Psi \}^B \bigr )$ is an element of $\S\left(\X;C_0^\infty(\O)\right)$ for all $\hbar \in (0,1]$.

The most difficult part of the proof is to show that the $\hbar$-dependent $C^*$-norm of the remainder
$R^{\diamond,2}_\hb(\Phi,\Psi)$ can be uniformly bounded in $\hbar$. The first ingredient
is the fact that the $\hbar$-dependent $C^*$-norm of the twisted crossed product is dominated by the
$L^1 \bigl ( \mathcal{X} ; \A \bigr )$-norm for all values of $\hbar \in (0,1]$,
$$
\snorm{\Phi}^B_{\hbar} \leq \snorm{\Phi}_{L^1},\ \ \ \ \
\forall \Phi \in \S\left(\X;C_0^\infty(\O)\right) \subset L^1 \bigl ( \mathscr{X};C_0(\Omega) \bigr )\subset \mathfrak C^B_\hb.
$$
Hence, if we can find $\hbar$-independent $L^1$ bounds on each term of the remainder, we have also estimated the
$\hbar$-dependent $C^*$-norm uniformly in $\hbar$.

There are four distinct types of terms in the remainder. Let us start with the first: we define
$$\bigl ( R^{\diamond,2}_{\hb,1}(\Phi,\Psi) \bigr )(x) := \int_{\mathcal{X}} \dd y \, \bigl
( R^{\theta,2}_{\hb,y - x}(\Phi) \bigr )(y) \, \theta_{\frac{\hbar}{2} y} [\Psi(x-y)] \,
e^{- i \hbar \Lambda^B_{\hbar}(x,y)}.$$
Then we have
$$\bnorm{R^{\diamond,2}_{\hb,1}(\Phi,\Psi)}_{\hbar}^B \leq
\bnorm{R^{\diamond,2}_{\hb,1}(\Phi,\Psi)}_{L^1(\mathcal{X};\A)}
$$
$$\leq \int_{\mathcal{X}} \dd x \int_{\mathcal{X}} \dd y \, \bnorm{\bigl ( R^{\theta,2}_{\hb,y - x}(\Phi)
\bigr )(y)}_{\A} \, \bnorm{\theta_{\frac{\hbar}{2} y} [\Psi(x-y)]}_{\A} \, \bnorm{e^{- i \hbar
\Lambda^B_{\hbar}(x,y)}}_{\A}
$$
$$= \int_{\mathcal{X}} \dd x \int_{\mathcal{X}} \dd y \, \bnorm{\bigl ( R^{\theta,2}_{\hb,y - x}(\Phi)
\bigr )(y)}_{\A} \, \bnorm{\Psi(x-y)}_{\A}
$$
$$= \int_{\mathcal{X}} \dd x \int_{\mathcal{X}} \dd y \, \bnorm{\bigl ( R^{\theta,2}_{\hb,- x}(\Phi)
\bigr )(y)}_{\A} \, \bnorm{\Psi(x)}_{\A}.$$

We inspect $\bnorm{\bigl ( R^{\theta,2}_{\hb,- x}(\Phi) \bigr )(y)}_{\A}$ more closely:
$$
\bnorm{\bigl ( R^{\theta,2}_{\hb,- x}(\Phi) \bigr )(y)}_{\A} \leq \frac{1}{4} \int_0^1 \dd \tau
\babs{(-x_j)(-x_k)} \bnorm{\theta_{- \tau \frac{\hbar}{2} x}[(\delta_j \delta_k \Phi)(y)]}_{\A}
= \frac{1}{8} \sabs{x_j  x_k} \bnorm{(\delta_j \delta_k \Phi)(y)}_{\A}.
$$
If we plug that back into the estimate of the $L^1$ norm, we get
$$
\bnorm{R^{\diamond,2}_{\hb,1}(\Phi,\Psi)}^B_{\hbar} \leq \frac{1}{8} \int_{\mathcal{X}} \dd x \int_{\mathcal{X}}
\dd y  \bnorm{(\delta_j \delta_k \Phi)(y)}_{\A} \bnorm{(Q_j Q_k \Psi)(x)}_{\A}
= \frac{1}{8} \bnorm{\delta_j \delta_k \Phi}_{L^1} \bnorm{Q_j Q_k \Psi}_{L^1}.
$$
The right-hand side is finite by the definition of $\S\left(\X;C_0^\infty(\O)\right)$. Similarly, the second term can be
estimated, just the roles of $\Phi$ and $\Psi$ are reversed.

Now to the second type of term: we define
\begin{equation*}
\bigl ( R^{\diamond,2}_{\hb,3}(\Phi,\Psi) \bigr )(x) := - \frac{1}{2} \int_{\mathcal{X}} \dd y (\delta_j \Phi)(y)
 \bigl ( Q_j R^{\theta,1}_{\hb,y} (\Psi) \bigr )(x-y) \, e^{- i \hbar \Lambda^B_{\hbar}(x,y)}
\end{equation*}
and estimate
\begin{equation*}
2 \bnorm{R^{\diamond,2}_{\hb,3}(\Phi,\Psi)}_{\hbar}^B \leq 2 \bnorm{R^{\diamond,2}_{\hb,3}(\Phi,\Psi)}_{L^1}
\leq \int_{\mathcal{X}} \dd x \int_{\mathcal{X}} \dd y  \bnorm{(\delta_j \Phi)(y)}_{\A} \, \bnorm{\bigl
( Q_j R^{\theta,1}_{\hb,y}(\Psi) \bigr )(x )}_{\A} .
\end{equation*}
The last factor needs to be estimated by hand:
$$
\bnorm{\bigl ( Q_j R^{\theta,1}_{\hb,y}(\Psi) \bigr )(x)}_{\A}\leq \frac{1}{2} \int_0^1 \dd \tau \, \sabs{x_j \,
y_k} \, \bnorm{\theta_{\tau \frac{\hbar}{2} y}[(\delta_k \Psi)(x)]}_{\A}
= \frac{1}{2} \sabs{x_j \, y_k} \, \bnorm{(\delta_k \Psi)(x)}_{\A}.
$$
This leads to the bound
\begin{equation*}
\bnorm{R^{\diamond,2}_{\hb,3}(\Phi,\Psi)}_{\hbar}^B \leq \frac{1}{4} \int_{\mathcal{X}} \dd x \int_{\mathcal{X}}
\dd y \, \bnorm{(\delta_j \Phi)(y)}_{\A} \, \sabs{x_j \, y_k} \, \bnorm{(\delta_k \Psi)(x)}_{\A}
= \frac{1}{4} \bnorm{Q_k \delta_j \Phi}_{L^1} \, \bnorm{Q_j \delta_k \Psi}_{L^1}.
\end{equation*}
The right-hand side is again finite since $\Phi , \Psi \in \S\left(\X;C_0^\infty(\O)\right)$ does not depend on
$\hbar$.

Estimating the two magnetic terms is indeed a bit more involved: we define
$$\bigl ( R^{\diamond,2}_{\hb,4}(\Phi,\Psi) \bigr )(x) := \frac{1}{2} \int_{\mathcal{X}} \dd y \, (Q_j \Phi)(y) \,
(\delta_j \Psi)(x-y) \, R^{\kappa,1}_\hb(x,y).$$
The usual arguments show the $C^*$-norm can be estimated by
$$
\bnorm{R^{\diamond,2}_{\hb,4}(\Phi,\Psi)}_{\hbar}^B \leq \int_{\mathcal{X}} \dd x \int_{\mathcal{X}} \dd y \,
\bnorm{(Q_j \Phi)(y)}_{\A} \, \bnorm{(\delta_j \Psi)(x-y)}_{\A} \, \bnorm{R^{\kappa,1}_\hb(x,y)}_{\A}
$$
which warrants a closer inspection of the last term: first of all, we note that
$$
R^{\kappa,1}_\hb(x,y) = \int_0^1 \dd \tau \, \frac{\dd}{\dd \epsilon} \bigl ( e^{- i \epsilon
\Lambda^B_{\epsilon}(x,y)} \bigr ) \bigg \vert_{\epsilon = \tau \hbar}
= \int_0^1 \dd \tau \, \Bigl ( - i \Lambda^B_{\epsilon}(x,y) - i \epsilon \, \tfrac{\dd}{\dd \epsilon}
\Lambda^B_{\epsilon}(x,y) \Bigr ) \, e^{- i \epsilon \Lambda^B_{\epsilon}(x,y)} \Big \vert_{\epsilon = \tau \hbar}.
$$
If we use Lemma~\ref{mag_twistedXprod:lem:flux_estimates} and $\hbar \leq 1$,  this leads to the following
norm estimate of $R^{\kappa,1}_\hb(x,y)$:
$$
\bnorm{R^{\kappa,1}_\hb(x,y)}_{\A} \leq \int_0^1 \dd \tau \, \Bigl (
\bnorm{\Lambda^B_{\tau \hbar}(x,y)}_{\A} + \hbar \tau \, \norm{\tfrac{\dd}{\dd \epsilon} \Lambda^B_{\epsilon}(x,y)
\Big \vert_{\epsilon = \tau \hbar}}_{\A} \Bigr ) \, \bnorm{e^{- i \epsilon \Lambda^B_{\tau \hbar}(x,y)}}_{\A}
$$
$$\leq \bnorm{B^{jk}}_{\A} \, \sabs{y_j} \, \sabs{x_k - y_k} + \tfrac{1}{2} \bnorm{\delta_l B^{jk}}_{\A} \,
\sabs{y_j} \, \sabs{x_k - y_k} \, \bigl ( \sabs{x_l - y_l} + \sabs{y_l} \bigr ).$$
Put together, this allows us to estimate the norm of $R^{\diamond,2}_{\hb,4}$ by
$$
\bnorm{R^{\diamond,2}_{\hb,4}(\Phi,\Psi)}_{\hbar}^B \leq \bnorm{B^{mk}}_{\A} \, \snorm{Q_j Q_m \Phi}_{L^1} \,
\snorm{Q_k \delta_j \Psi}_{L^1}+ $$
$$+ \tfrac{1}{2} \bnorm{\delta_l B^{mk}}_{\A} \, \Bigl ( \snorm{Q_j Q_m Q_l \Phi}_{L^1} \, \snorm{Q_k \delta_j
\Psi}_{L^1}+ \snorm{Q_j Q_m \Phi}_{L^1} \, \snorm{Q_k Q_l \delta_j \Psi}_{L^1} \Bigr ).$$
Now on to the last term,
$$\bigl ( R^{\diamond,2}_{\hb,6}(\Phi,\Psi) \bigr )(x) := \int_{\mathcal{X}} \dd y \, \Phi(y) \, \Psi(x-y) \,
R^{\kappa,2}_\hb(x,y).$$
Using the explicit form of $R^{\kappa,2}_\hb(x,y)$,
$$
R^{\kappa,2}_\hb(x,y) = \int_0^1 \dd \tau \, (1 - \tau) \, \frac{\dd^2}{\dd \epsilon^2} \bigl ( e^{- i
\epsilon \Lambda^B_{\epsilon}(x,y)} \bigr ) \bigg \vert_{\epsilon = \tau \hbar}$$
$$= \int_0^1 \dd \tau \, (1 - \tau) \, \Bigl [ - i 2 \tfrac{\dd}{\dd \epsilon} \Lambda^B_{\epsilon}(x,y) - i
\epsilon \tfrac{\dd^2}{\dd \epsilon^2} \Lambda^B_{\epsilon}(x,y) - \Bigl ( \Lambda^B_{\epsilon}(x,y) + \epsilon
\, \tfrac{\dd}{\dd \epsilon} \Lambda^B_{\epsilon}(x,y) \Bigr )^2 \Bigr ] \Big \vert_{\epsilon = \tau \hbar} \,
e^{- i \tau \hbar \Lambda^B_{\tau \hbar}(x,y)},$$
in conjunction with the estimates from Lemma~\ref{mag_twistedXprod:lem:flux_estimates} (which are uniform in $\tau$),
we get
$$\bnorm{R^{\kappa,2}_\hb(x,y)}_{\A} \leq \int_0^1 \dd \tau \, (1 - \tau) \,
\left [ 2 \bnorm{\tfrac{\dd}{\dd \epsilon} \Lambda^B_{\epsilon}(x,y)
\big \vert_{\epsilon = \tau \hbar}}_{\A} + \tau \bnorm{\tfrac{\dd^2}{\dd \epsilon^2}
\Lambda^B_{\epsilon}(x,y) \big \vert_{\epsilon = \tau \hbar}}_{\A}
+ \right . $$
$$ \left .
+ \Bigl ( \bnorm{\Lambda^B_{\hbar \tau}(x,y)}_{\A} + \tau \, \bnorm{\tfrac{\dd}{\dd \epsilon}
\Lambda^B_{\epsilon}(x,y) \big \vert_{\epsilon = \tau \hbar}} \Bigr )^2 \right ] \,
\bnorm{e^{- i \tau \hbar \Lambda^B_{\tau \hbar}(x,y)}}_{\A}$$
$$= \int_0^1 \dd \tau \, (1 - \tau) \, \left [
2 \bnorm{\tfrac{\dd}{\dd \epsilon} \Lambda^B_{\epsilon}(x,y) \big \vert_{\epsilon = \tau \hbar}}_{\A} +
\tau \bnorm{\tfrac{\dd^2}{\dd \epsilon^2} \Lambda^B_{\epsilon}(x,y) \big \vert_{\epsilon = \tau \hbar}}_{\A}
+ \right .$$
$$ \left .
+ \bnorm{\Lambda^B_{\hbar \tau}(x,y)}_{\A}^2 + 2 \tau \, \bnorm{\Lambda^B_{\hbar \tau}(x,y)}_{\A} \,
\bnorm{\tfrac{\dd}{\dd \epsilon} \Lambda^B_{\epsilon}(x,y) \big \vert_{\epsilon = \tau \hbar}}_{\A} +
\tau^2 \bnorm{\tfrac{\dd}{\dd \epsilon} \Lambda^B_{\epsilon}(x,y) \big \vert_{\epsilon = \tau \hbar}}_{\A}^2
\right ].
$$
Hence, we can bound the $\hbar$-dependent $C^*$-norm of $R^{\diamond,2}_{\hb,6}$ by
$$
\bnorm{R^{\diamond,2}_{\hb,6}}_{\hbar}^B \leq \bnorm{\delta_l B^{jk}}_{\A}  \Bigl ( \bnorm{Q_j \Phi}_{L^1}
\bnorm{Q_k Q_l \Psi}_{L^1} + \bnorm{Q_j Q_l \Phi}_{L^1} \, \bnorm{Q_k \Psi}_{L^1} \Bigr )+ $$
$$+ \tfrac{1}{6} \bnorm{\delta_l \delta_m B^{jk}}_{\A} \Bigl ( \bnorm{Q_j \Phi}_{L^1}
\bnorm{Q_k Q_l Q_m \Psi}_{L^1} \!+ \!\bnorm{Q_j Q_m \Phi}_{L^1} \bnorm{Q_k Q_l \Psi}_{L^1}\! +
\bnorm{Q_j Q_l Q_m \Phi}_{L^1} \bnorm{Q_k \Psi}_{L^1} \Bigr )+$$
$$+ \tfrac{1}{2} \bnorm{B^{jk}}_{\A}  \bnorm{B^{j'k'}}_{\A}  \bnorm{Q_j Q_{j'} \Phi}_{L^1}
\bnorm{Q_k Q_{k'} \Psi}_{L^1}+$$
$$+ \tfrac{1}{3} \bnorm{B^{jk}}_{\A} \, \bnorm{\delta_{l'} B^{j'k'}}_{\A}  \Bigl ( \bnorm{Q_j Q_{j'}
\Phi}_{L^1}  \bnorm{Q_k Q_{k'} Q_{l'} \Psi}_{L^1} + \bnorm{Q_j Q_{j'} Q_{l'} \Phi}_{L^1}  \bnorm{Q_k Q_{k'}
\Psi}_{L^1} \Bigr )+$$
$$+ \tfrac{1}{12} \bnorm{\delta_l B^{jk}}_{\A}  \bnorm{\delta_{l'} B^{j'k'}}_{\A}  \Bigl ( \bnorm{Q_j Q_{j'}
\Phi}_{L^1} \bnorm{Q_k Q_{k'} Q_l Q_{l'}\Psi}_{L^1}+$$
$$+ 2 \bnorm{Q_j Q_{j'} Q_l \Phi}_{L^1} \bnorm{Q_k Q_{k'} Q_{l'} \Psi}_{L^1}+ \bnorm{Q_j Q_{j'} Q_l Q_{l'}
\Phi}_{L^1} \, \bnorm{Q_k Q_{k'} \Psi}_{L^1} \Bigr ).$$

Putting all these individual estimates together yields a bound on $\bnorm{R^{\diamond,2}_\hb(\Phi,\Psi)}_{\hbar}^B$
which is uniform in $\hbar$ and the proof of the Theorem is finished.
\end{proof}
\begin{corollary}\label{corollary:asymptotic_expansion_product}
Assume the components of $B$ are in $BC^{\infty}(\Omega)$. Let $f , g \in \mathcal{S} \bigl ( \mathscr{X}^* ; C_0^\infty(\O)
\bigr )$ and $\hbar \in (0,1]$. Then the product $f \sharp^B_{\hbar} g$ can be expanded in powers of $\hbar$,
\begin{equation}
f \sharp^B_{\hbar} g = f \, g - \hbar \tfrac{i}{2} \{ f , g \}_B + \hb^2 R^{\sharp,2}_\hb(f,g),
\end{equation}
where $fg$ is the pointwise product and $\{ f , g \}_B$ is the magnetic Poisson bracket defined as in
equation (\ref{oisson}). All terms are in
$\mathcal{S} \bigl ( \mathscr{X}^* ; C_0^\infty(\O)\bigr )$ and the remainder satisfies
$\parallel R^{\sharp,2}_\hb(f,g)\parallel_{\mathfrak B^B_\hb}\,\le C$ uniformly in $\hbar$.
\end{corollary}

\begin{proof}
The proof follows from equations \eqref{novela}, \eqref{bunica},  \eqref{bunika} and
Theorem~\ref{theorem:asymptotic_expansion_product},
keeping in mind that the partial Fourier transforms are isomorphisms $\mathcal{S} \bigl ( \mathscr{X}^* ; C^\infty(\O)
\bigr ) \overset{\longrightarrow}{\longleftarrow} \mathcal{S} \bigl ( \mathscr{X} ; C^\infty(\O) \bigr )$ that extend
to automorphisms between the $C^*$-algebras $\mathfrak B^B_\hb$ and $\mathfrak C^B_\hb$.
\end{proof}

\section{Strict deformation quantization}\label{defoq}

To make this precise, we repeat an already standard concept.
For more details and motivation, the reader could see \cite{Rie1,Rie2,La} and references therein.

\begin{definition}
Let $(\mathcal S,\circ,\{\cdot , \cdot\})$ be a real Poisson algebra which is densely contained on the selfadjoint part
$\CCC_{0,\R}$ of an abelian $C^*$-algebra $\CCC_0$. {\rm A strict deformation quantization of the Poisson algebra}
$\mathcal S$ is a family of $\R$-linear injections $\bigl ( \mathfrak{Q}_\hb:\mathcal
S\rightarrow \CCC_{\hb,\R} \bigr )_{\hb\in I}$, where $I\subset \R$ contains $0$ as an accumulation point,
$\CCC_{\hb,\R}$ is the selfadjoint part of the $C^*$-algebra $\CCC_\hb$, with products and norms denoted by
$\diamond_\hb$ and $\snorm{\cdot}_\hb$, $\mathfrak{Q}_0$ is just the inclusion map and $\mathfrak{Q}_\hb(\mathcal S)$
is a subalgebra of $\CCC_{\hb,\R}$.

The following conditions are required for each $\Phi,\Psi\in\mathcal S$
\begin{enumerate}[(i)]
\item
{\rm Rieffel axiom}: the mapping $I \ni \hb \mapsto \bnorm{\mathfrak{Q}_\hb(\Phi)}_{\hb}$ is continuous.
\item
{\rm Von Neumann axiom}:
$$
\lim_{\hb \rightarrow 0} \bnorm{\tfrac{1}{2} \left[\mathfrak{Q}_\hb(\Phi)\diamond_{\hb}\mathfrak{Q}_\hb(\Psi) +
\mathfrak{Q}_\hb(\Psi)\diamond_{\hbar}\mathfrak{Q}_\hb(\Phi)\right]-\mathfrak{Q}_\hb(\Phi\circ\Psi)}_{\hb} = 0.
$$
\item
{\rm Dirac axiom}:
$$
\lim_{\hb \rightarrow 0} \Bnorm{\tfrac{i}{\hb}\left[\mathfrak{Q}_\hb(\Phi)\diamond_{\hb}\mathfrak{Q}_\hb(\Psi)
-\mathfrak{Q}_\hb(\Psi)\diamond_{\hbar}\mathcal{Q}_\hb(\Phi) \right] - \mathfrak{Q}_\hb(\{\Phi,\Psi\}) }_{\hb} = 0.
$$
\end{enumerate}
\end{definition}

Putting this into the present context, we have

\begin{theorem}\label{pastur}
Assume that $B^{jk} \in BC^\infty(\O)$ and $I=[0,1]$. Then the family of injections
$$
\left(\,\S\left(\X,C_0^\infty(\O)\right)_\R\hookrightarrow \CCC^{B}_{\hb,\R}\,\right)_{\hb\in I}
$$
defines a strict deformation quantization.
\end{theorem}

\begin{proof}
By Proposition~\ref{pesti} and Proposition~\ref{intrucul}, $\mathcal{S} \bigl ( \X , C^{\infty}_0(\Omega) \bigr )_\R$
can be seen a Poisson algebra with respect to $\diamond_0$ and $\{ \cdot , \cdot \}^B$ as well as a subalgebra of the real part
of each of the twisted crossed product $\mathfrak{C}^B_{\hbar}$.

Von Neumann and Dirac axioms are direct consequences of Theorem~\ref{theorem:asymptotic_expansion_product}.

The Rieffel axiom can be checked exactly as in \cite{MP2}, which builds on results from \cite{Ni,Rie}. The fact that
the algebra $\A$ in \cite{MP2} consisted of continuous functions defined on the group $\X$ itself does not play any role here.
\end{proof}

A partial Fourier transform transfers these results directly to $\mathcal{S} \bigl ( \X^* , C_0^{\infty}(\Omega) \bigr )$
and $\BB^B_{\hbar}$, objects which are  natural in the context of Weyl calculus. In this way we extend the main result of
\cite{MP2} to magnetic twisted actions on {\it general} abelian $C^*$-algebras.

\begin{corollary}
Assume that $B^{jk} \in C^\infty(\O)$. Let $I=[0,1]$. Then the family of injections
$$
\Bigl ( \S \bigl ( \X^* , C_0^{\infty}(\O) \bigr )_\R \hookrightarrow \BB^{B}_{\hb,\R} \Bigr )_{\hbar \in I}
$$
defines a strict deformation quantization, where the Poisson  structure in  $\mathcal{S}
\bigl ( \X^*,C^{\infty}_0(\O) \bigr )_\R$ is given by point-wise multiplication and the Poisson bracket
$\{ \cdot , \cdot \}_B$.
\end{corollary}
\begin{proof}
The proof is straightforward from the Corollary~\ref{corollary:asymptotic_expansion_product} and the above theorem,
after noticing that the partial Fourier transform is an isomorphism between the Poisson algebras
$\mathcal{S} \bigl ( \mathscr{X}^* ; C_0^\infty(\O) \bigr )$ and $\mathcal{S} \bigl ( \mathscr{X} ; C_0^\infty(\O)
\bigr )$, and it extends to an isomorphisms between the $C^*$-algebras $\mathfrak B^B_\hb$ and $\mathfrak C^B_\hb$.
\end{proof}

\section*{Appendix: Estimates on the magnetic flux}\label{defrroq}

In the next lemma we gather some useful estimates on the scaled magnetic flux and its exponential, that are used in the
proofs of Propositions \ref{intrucul} and \ref{theorem:asymptotic_expansion_product}.
\begin{lemma}\label{mag_twistedXprod:lem:flux_estimates}
Assume the components of $B$ are in $BC^{\infty}(\Omega)$ and $\hbar \in (0,1]$.
\begin{enumerate}[(i)]
\item For all $a , \alpha \in \N^n$ there exist constants $C^j > 0$, $C^{jk} > 0$, $j , k \in \{ 1 ,
\ldots , n \}$, depending on $B^{jk}$ and its $\delta$-derivatives up to $(\abs{a} + \abs{\alpha})$th order, such that
\begin{equation*}
\bnorm{\partial_x^a \delta^{\alpha} \Lambda^B_{\hbar}(x,y)}_{\A} \leq \sum_{j = 1}^n C_1^j \, \sabs{y_j} +
\sum_{j , k = 1}^n C_2^{jk} \, \sabs{y_j} \, \sabs{x_k - y_k}.
\end{equation*}
\item For all $a , \alpha \in \N^n$ there exists a polynomial $p_{a \alpha}$ in $2n$ variables, with coefficients
$K_{bc} \geq 0$, such that

$$\bnorm{\partial_x^a \delta^{\alpha} e^{- i \hbar \Lambda^B_{\hbar}(x,y)}}_{\A} \leq p_{a \alpha} \bigl ( \sabs{y_1} ,
\ldots , \sabs{y_n} , \sabs{x_1 - y_1} , \ldots , \sabs{x_n - y_n} \bigr )$$
$$= \sum_{\sabs{b} + \sabs{c} \leq 2 (\sabs{a} + \sabs{\alpha})} K_{bc} \, \sabs{y^b} \, \sabs{(x - y)^c}.$$

\item The following estimates which are uniform in $\hbar$ and $\tau$ hold :
$$\bnorm{\Lambda^B_{\hbar \tau}(x,y)}_{\A} \leq \sum_{jk}\,\bnorm{B^{jk}}_{\A} \, \sabs{y_j} \, \sabs{x_k - y_k},$$

$$\Bnorm{\tfrac{\dd}{\dd \epsilon} \Lambda^B_{\epsilon}(x,y) \big \vert_{\epsilon = \tau \hbar}}_{\A}
\leq \sum_{jkl}\,\bnorm{\delta_l B^{jk}}_{\A} \, \sabs{y_j} \, \sabs{x_k - y_k} \, \bigl ( \sabs{x_l - y_l} +
\sabs{y_l} \bigr ),$$

$$\Bnorm{\tfrac{\dd^2}{\dd \epsilon^2} \Lambda^B_{\epsilon}(x,y) \big \vert_{\epsilon = \tau \hbar}}_{\A}$$
$$\leq \sum_{jklm}\bnorm{\delta_l \delta_m B^{jk}}_{\A} \sabs{y_j}  \sabs{x_k - y_k}\bigl
( \sabs{x_l - y_l} \sabs{x_m - y_m} + \sabs{y_l} \sabs{x_m - y_m} + \sabs{y_l}  \sabs{y_m} \bigr ).$$

\end{enumerate}
\end{lemma}

\begin{proof}
    (i) and (ii) follow directly from the explicit parametrization of
    the magnetic flux.
\medskip

\noindent
(iii) Throughout the proof we are going to use Einstein's summation convention, \ie\ repeated indices in a product are
    summed over from $1$ to $\mathrm{dim}(\X)$. From the explicit parametrization \eqref{lam}

\begin{equation*}
\Lambda^B_{\epsilon}(x,y) = y_j \, (x_k - y_k) \, \int_0^1 \dd t \int_0^t \dd s \, \theta_{\epsilon (s -
\nicefrac{1}{2}) x + \epsilon (t - s) y}[B^{jk}],
\end{equation*}

we compute first and second derivative of $\Lambda^B_{\epsilon}(x,y)$ with respect to $\epsilon$, using dominated
convergence to interchange differentiation with respect to the parameter $\epsilon$ and integration with respect to
$t$ and $s$,
$$
\frac{\dd}{\dd \epsilon} \Lambda^B_{\epsilon}(x,y) = y_j \, (x_k - y_k) \, \int_0^1 \dd t \int_0^t \dd s \,
\bigl ( s (x_l - y_l) + t y_l - \tfrac{1}{2} x_l \bigr ) \, \theta_{\epsilon (s - \nicefrac{1}{2}) x + \epsilon
(t - s) y}[\delta_l B^{jk}],$$

$$\frac{\dd^2}{\dd \epsilon^2} \Lambda^B_{\epsilon}(x,y) = y_j \, (x_k - y_k) \, \int_0^1 \dd t \int_0^t \dd s \,
\bigl ( s (x_l - y_l) + t y_l - \tfrac{1}{2} x_l \bigr )  \bigl ( s (x_m - y_m) + t y_m - \tfrac{1}{2} x_m \bigr )
\cdot $$
$$ \cdot
\theta_{\epsilon (s - \nicefrac{1}{2}) x + \epsilon (t - s) y}[\delta_l \delta_m B^{jk}].$$

The estimate on the flux itself follows from the fact that all the automorphisms $\theta_z$ are isometric in $\A$:
$$
\bnorm{\Lambda^B_{\tau \hbar}(x,y)}_{\A} \leq \sabs{y_j} \, \sabs{x_k - y_k} \, \int_0^1 \dd t \int_0^t \dd s \,
\bnorm{\theta_{\epsilon (s - \nicefrac{1}{2}) x + \epsilon (t - s) y}[B^{jk}]}_{\A}
\leq \bnorm{B^{jk}}_{\A} \, \sabs{y_j} \, \sabs{x_k - y_k}.$$

Using the triangle inequality to estimate $\sabs{x_l}$ from above by $\sabs{x_l - y_l} + \sabs{y_l}$, we get

$$\norm{\frac{\dd}{\dd \epsilon} \Lambda^B_{\epsilon}(x,y) \bigg \vert_{\epsilon = \tau \hbar}}_{\A} $$
$$\leq \sabs{y_j}  \sabs{x_k - y_k} \, \int_0^1 \dd t \int_0^t \dd s \, \bigl ( s \sabs{x_l - y_l} + t \sabs{y_l} +
\tfrac{1}{2} \sabs{x_l} \bigr ) \, \bnorm{\theta_{\tau \hbar (s - \nicefrac{1}{2}) x + \tau \hbar (t - s)
y}[\delta_l B^{jk}]}_{\A}$$
$$= \bnorm{\delta_l B^{jk}}_{\A} \sabs{y_j} \, \sabs{x_k - y_k} \int_0^1 \dd t \int_0^t \dd s
\bigl ( s \sabs{x_l - y_l} + t \sabs{y_l} + \tfrac{1}{2} \sabs{x_l} \bigr )$$

$$\leq \bnorm{\delta_l B^{jk}}_{\A} \, \sabs{y_j} \, \sabs{x_k - y_k} \, \bigl ( \sabs{x_l - y_l} + \sabs{y_l}
\bigr ).$$

In a similar fashion, we obtain the estimate for the second-order derivative,

$$\norm{\frac{\dd^2}{\dd \epsilon^2} \Lambda^B_{\epsilon}(x,y) \bigg \vert_{\epsilon = \tau \hbar}}_{\A}$$
$$\leq \sabs{y_j} \, \sabs{x_k - y_k} \, \int_0^1 \dd t \int_0^t \dd s \, \Babs{\bigl ( s (x_l - y_l) + t y_l -
\tfrac{1}{2} x_l \bigr ) \, \bigl ( s (x_m - y_m) + t y_m - \tfrac{1}{2} x_m \bigr )}\cdot $$
$$ \cdot\bnorm{\theta_{\tau \hbar (s - \nicefrac{1}{2}) x + \tau \hbar (t - s) y}[\delta_l \delta_m B^{jk}]}_{\A}$$
$$\leq \bnorm{\delta_l \delta_m B^{jk}}_{\A} \, \sabs{y_j} \, \sabs{x_k - y_k} \, \int_0^1 \dd t \int_0^t \dd s
\, \Bigl ( s^2 \sabs{x_l - y_l} \, \sabs{x_m - y_m} + 2 st \sabs{y_l} \, \sabs{x_m - y_m}
+ \Bigr .$$
$$ \Bigl . + s \sabs{x_l - y_l} \, \sabs{x_m} + t \sabs{y_l} \, \sabs{x_m} + t^2 \sabs{y_l} \, \sabs{y_m} + \tfrac{1}{4}
\sabs{x_l} \, \sabs{x_m} \Bigr )$$
$$\leq \bnorm{\delta_l \delta_m B^{jk}}_{\A} \, \sabs{y_j} \, \sabs{x_k - y_k} \, \bigl (
\sabs{x_l - y_l} \, \sabs{x_m - y_m} + \sabs{y_l} \, \sabs{x_m - y_m} + \sabs{y_l} \, \sabs{y_m}
\bigr ).$$

This finishes the proof.
\end{proof}

{\bf Acknowledgements:}
F. Belmonte is supported by {\it N\'ucleo Cientifico ICM P07-027-F "Mathematical
Theory of Quantum and Classical Magnetic Systems"}.
M. Lein is supported by Chilean Science Foundation {\it Fondecyt} under the Grant 1090008.
M. M\u antoiu is supported by {\it N\'ucleo Cientifico ICM P07-027-F "Mathematical
Theory of Quantum and Classical Magnetic Systems"} and by Chilean Science Foundation {\it Fondecyt} under the Grant
1085162. He thanks Serge Richard and Rafael Tiedra de Aldecoa for their interest in this project.
Part of this article has been written while the three authors were participating to the program
{\it Spectral and Dynamical Properties of Quantum Hamiltonians}.
They are grateful to the Centre Interfacultaire Bernoulli for the excellent atmosphere and conditions.


\end{document}